\documentclass[11pt]{myclass}

\pdfoutput=1

\usepackage{euscript}
\usepackage{picins}
\newcommand{\eps}{\varepsilon}

\renewcommand{\c}[1]{\ensuremath{\EuScript{#1}}}
\renewcommand{\b}[1]{\ensuremath{\mathbb{#1}}}
\renewcommand{\sf}[1]{\textsf{#1}}
\newcommand{\B}[1]{\ensuremath{\textbf{#1}}}
\renewcommand{\P}{\ensuremath{\B{Pr}}}

\newcommand{\sip}{\textsf{sip }}
\newcommand{\IP}[2]{\ensuremath{ \langle #1 , #2 \rangle}}
\newcommand{\wid}{\omega}
\usepackage{graphics}

\newcommand\PPic[2]{
    \begin{minipage}{#1}%
        \hspace{-0.4cm}
        \makebox[0cm][l]{\includegraphics{#2}}
    \end{minipage}}

\newlength{\ppicwd}

\usepackage {graphicx}
\usepackage {subfigure}

\newcommand {\subplaatje} [2]
{
  \hfill
  \subfigure []
  {
    \includegraphics [#1] {figures/#2}
    \label {#2}
  }
}

\newcommand {\plaatjes} [3]
{
  \begin {figure} [tb]
    \centering
    #1
    \hfill
    \mbox {}
    \caption {#2}
    \label {#3}
  \end {figure}
}

\newcommand {\eenplaatje} [3] []
{
  \begin {figure} [tb]
    \centering
    \includegraphics [#1] {figures/#2}
    \caption {#3} \label {#2}
  \end {figure}
}

\newcommand {\tweeplaatjes} [4] []
{
  \plaatjes
  {
    \subplaatje {#1} {#2}
    \subplaatje {#1} {#3}
  } {#4} {#2+#3}
}

\newcommand {\drieplaatjes} [5] []
{
  \plaatjes
  {
    \subplaatje {#1} {#2}
    \subplaatje {#1} {#3}
    \subplaatje {#1} {#4}
  } {#5} {#2+#3+#4}
}

\newcommand {\vierplaatjes} [6] []
{
  \plaatjes
  {
    \subplaatje {#1} {#2}
    \subplaatje {#1} {#3}
    \subplaatje {#1} {#4}
    \subplaatje {#1} {#5}
  } {#6} {#2+#3+#4+#5}
}

\title{Shape Fitting on Point Sets with Probability Distributions}
\author{%
	\begin{tabular}{c}
	Maarten L\"offler\\
	\small Utrecht University\\
	\small loffler@cs.uu.nl
	\end{tabular}
	\and 
	\begin{tabular}{c}
	Jeff Phillips\\
	\small Duke University\\
	\small jeffp@cs.duke.edu
	\end{tabular}
}

\begin{document}
\begin{titlepage}
\maketitle

\begin{center}  \today  \end{center}

\begin{abstract} 
A typical computational geometry problem begins: Consider a set $P$ of $n$ points in $\b{R}^d$. However, many applications today work with input that is not precisely known, for example when the data is sensed and has some known error model. What if we do not know the set $P$ exactly, but rather we have a probability distribution $\mu_p$ governing the location of each point $p \in P$?

Consider a set of (non-fixed) points $P$, and let $\mu_P$ be the probability distribution of this set. We study several measures (e.g. the radius of the smallest enclosing ball, or the area of the smallest enclosing box) with respect to $\mu_P$. The solutions to these problems do not, as in the traditional case, consist of a single answer, but rather a distribution of answers. We hence describe a data structure, called an $\eps$-quantization, that can approximate such a distribution within $\eps$ in $O(1/\eps)$ space. We also extend this data structure to answer higher dimensional queries of $\mu_P$ (e.g. the length and width of the smallest enclosing box in $\b{R}^2$).

Rather than compute a new data structure for each measure we are interested in, we can also compute a single data structure that allows us to answer many questions at once. This data structure, an $(\eps, \alpha)$-kernel, is based on $\alpha$-kernel coresets and can be used to create approximate $\eps$-quantizations for geometric problems involving extent measures.

Thirdly, we introduce a data structure that can answer questions of the type `what is the probability that point $q$ is in the smallest enclosing ball of $P$?' For a given distribution $\mu_P$ and summarizing shape (e.g. the smallest enclosing ball), we define an $\eps$-shape inclusion probability function to be a function that assigns to a query point $q \in \b{R}^d$ a value that is at most $\eps$ away from the probability that $q$ is contained in this summarizing shape of $P$. This results in a probability description more directly linked to the space that the input points live in.

We provide simple and efficient randomized algorithms for computing all of these data structures, which are easy to implement and practical. We provide some experimental results to assert this. We also provide more involved deterministic algorithms for $\eps$-quantizations for problems involving shapes with bounded VC-dimension that run in time polynomial in $n$ and $1/\eps$.
\end{abstract}
\end{titlepage}

\section{Introduction}

The input for a typical computational geometry problem is a set $P$ of $n$ points in $\b{R}^2$, or more generally $\b{R}^d$.
Traditionally, such a set of points is assumed to be known exactly, and indeed, in the 1980s and 1990s such an assumption was often justified because much of the input data was hand-constructed for computer graphics or simulations.
However, in many modern applications the input is sensed from the real world, and such data is inherently imprecise.
Therefore, there is a growing need for methods that are able to deal with imprecision.

An early model to quantify imprecision in geometric data, motivated by finite precision of coordinates, is {\it $\eps$-geometry}, introduced by Guibas \etal~\cite {gss-egbra-89}. In this model, the input is given by a traditional point set $P$, where the imprecision is modeled by a single extra parameter $\eps$. The true point set is not known, but it is certain that for each point in $P$ there is a point in the disk of radius $\eps$ around it.
This model has proven fruitful and is still often used due to its simplicity.
To name a few, Guibas \etal~\cite {gss-cscah-93} define \emph {strongly convex} polygons: polygons that are guaranteed to stay convex, even when the vertices are perturbed by $\eps$. 
Bandyopadhyay and Snoeyink~\cite{bs-ads-04} compute the set of all potential simplices in $\b R^2$ and $\b R^3$ that could belong to the Delaunay triangulation.
Held and Mitchell~\cite {hm-ticpps-08} and L\"offler and Snoeyink~\cite {ls-dtip-08} study the problem of preprocessing a set of imprecise points under this model, so that when the true points are specified later some computation can be done faster.

A more involved model for imprecision can be obtained by not specifying a single $\eps$ for all the points, but allowing a different radius for each point, or even other shapes of imprecision regions. This allows for modeling imprecision that comes from different sources, independent imprecision in different dimensions of the input, etc. This extra freedom in modeling comes at the price of more involved algorithmic solutions, but still many results are available.
Nagai and Tokura~\cite {nt-teb-00} compute the union and intersection of all possible convex hulls to obtain bounds on any possible solution, as does Ostrovsky-Berman and Joskowicz~\cite {obj-ue-05} in a setting allowing some dependence between points.
Van Kreveld and L\"offler~\cite {kl-bgmips-06} study the problem of computing the smallest and largest possible values of several geometric extent measures, such as the diameter or the radius of the smallest enclosing ball, where the points are restricted to lie in given regions in the plane.
Kruger~\cite {k-bmips-08} extends some of these results to higher dimensions.

However, some applications dealing with sensed data provide more information about the imprecision than just a region, and a probability distribution governing the expected location of each point may be available.
In robotic mapping~\cite{EP04} careful error models are used to govern the laser range finder data.
In data mining~\cite{AS00} original data is often perturbed by a known model for privacy preserving purposes.
In databases~\cite{CKP03} large data sets may be summarized as probability distributions to store them more compactly.
The atoms of a protein structure have probabilistic distributions as determined by NMR spectroscopy reconstruction algorithms~\cite{PYCDB06}, rotamers, or other variability.
Similarly, probability distribution models are produced for GIS data, data from sensor networks, astrological data, and many other sources.
In these cases, the above threshold error models could be adapted to this data by choosing an error distance beyond which the probability is below a certain threshold. However, the solutions produced under the threshold error models depend heavily on the boundary cases of the error model, while it is reasonable to expect the points are more likely to appear near the ``center'' of the regions.
Working directly with probability distributions can provide more accurate answers to geometric questions about such sets of points.

This paper studies the computation of extent measures on uncertain point sets governed by probability distributions.
Unsurprisingly, directly using the probability distribution error model creates harder algorithmic problems, and many questions may be impossible to answer exactly under this model.  But since the data is imprecise to begin with, it is also reasonable to construct approximate answers.  Our algorithms have approximation guarantees with respect to the original distributions, not an approximation of them.
Instead of reinventing computational geometry for probability distributions, this paper reduces problems on data governed by probability distributions to discrete and well-studied computational geometry problems on precise point sets.

\subsection{Problem Statement}
Let $\mu_p : \b{R}^d \to \b{R}^+$ describe the probability distribution of a point $p$ where $\int_{x \in \b{R}^d} \mu_p(x) \; dx = 1$.
Let $\mu_P : \b{R}^d \times \b{R}^d \times \ldots \times \b{R}^d \to \b{R}^+$ describe the distribution of a point set $P$ by the joint probability over each $p \in P$.
For simplicity we refer to the space $\b{R}^d \times \ldots \times \b{R}^d$ as $\b{R}^{dn}$ when it is a product of $n$ $d$-dimensional spaces.
For this paper we will assume $\mu_P(q_1, q_2, \ldots, q_n) = \prod_{i=1}^n \mu_{p_i}(q_i)$, so the distribution for each point is independent, although for our randomized algorithms this restriction can be easily circumvented.

Given a distribution $\mu_P$ we can ask questions traditionally asked of point sets that are given precisely, instead of as distributions (e.g. the diameter or the axis-aligned bounding box). In the presence of imprecision, the answer to such a question if not a single value or structure, but also a \emph {distribution} of answers.
The point of this paper is not just how to answer geometric questions about these distributions, but how to concisely represent them.

\paragraph{$\eps$-Quantizations.}
Let $f : \b{R}^{dn} \to \b{R}$ be a single-valued function on a fixed point set, such as the radius of the minimum enclosing ball.
For a query value $v$,
$$
f^{\leq}_{\mu_P}(v) = \int_{Q \in \b{R}^{dn}} 1(f(Q) \leq v) \cdot \mu_P(Q) \, dQ,
$$
where $Q$ is taken over all size $n$ point sets in $\b{R}^d$ and $1(\cdot)$ is the indicator function, is the probability that $f$ will yield a value less than or equal to $v$, given the distribution $\mu_P$. Then $f^{\leq}_{\mu_P}$ is the cumulative density function of the distribution of possible values that $f$ can take.
Ideally, we would return the function $f^{\leq}_{\mu_P}$ so we could quickly answer any query exactly, however, it is not clear how to compute closed forms for such functions for one specific value, let alone all values.
Rather, we introduce a data structure, which we call an \emph{$\eps$-quantization}, to answer such queries approximately and efficiently.
For an isotonic function $f^{\leq}_{\mu_P}$ and any value $v$, an $\eps$-quantization, $R$, guarantees  that $|R(v) - f^{\leq}_{\mu_P}(v)| \leq \eps$.  Furthermore, the size of an $\eps$-quantization is always dependent only on $\eps$, not on $|P|$ or $\mu_P$.

Sometimes a statistic for a point set has multiple values, such as the width of the minimum enclosing axis-aligned rectangle along the $x$-axis and the $y$-axis.  For a function $f : \b{R}^{dn} \to \b{R}^k$ let
$$
f^{\preceq}_{\mu_P}(v_1, \ldots, v_k) = \int_{Q \in \b{R}^{dn}} 1(f(Q) \preceq (v_1, \ldots, v_k) ) \cdot \mu(Q) \, dQ,
$$
where for a point $p \in \b{R}^k$ the operation $p \preceq (v_1, \ldots, v_k)$ determines whether $p_i \leq v_i$ for each $i$, where $p_i$ is the $i$th coordinate of $p$.  Note that $f^{\preceq}_{\mu_P}$ must be isotonic in the sense that for two points $p,q \in \b{R}^{k}$ if $p \preceq q$ then $f^{\preceq}_{\mu_P}(p) \leq f^{\preceq}_{\mu_P}(q)$.
A \emph{$k$-variate $\eps$-quantization} $R$ for an isotonic function $f^{\preceq}_{\mu_P} : \b{R}^k \to [0,1]$ and for a query $v \in \b{R}^k$ guarantees  $|R(v) - f^{\preceq}_{\mu_P}(v)| \leq \eps$.  The size of a multivariate $\eps$-quantization is dependent only on $\eps$ and $k$.

\paragraph{$(\eps, \alpha)$-Kernels.}
Rather than compute a new data structure for each measure we are interested in, we can also compute a single data structure that allows us to answer many questions at once.
For an isotonic function $f_{\mu_P}^{\leq} : \b{R}^+ \to [0,1]$ an \emph{$(\eps, \alpha)$-quantization} $M$ guarantees that there exists a point $x^\prime$ such that
(1) $|x-x^\prime| \leq \alpha x$ and
(2) $|M(x) - f_{\mu_P}^{\leq}(x^\prime)| \leq \eps$.
An \emph{$(\eps, \alpha)$-kernel} is a data structure that can produce an $(\eps, \alpha)$-quantization for $f_{\mu_P}^{\leq}$ where $f$ measures the width in any directions and whose size depends only on $\frac{1}{\eps}$ and $\frac{1}{\alpha}$.

\paragraph{Shape Inclusion Probabilities.}
To summarize a point set $P \subset \b{R}^d$, we often approximate it with a shape, such as the smallest enclosing ball.  For $k$-variate $\eps$-quantizations with large $k$, it can be hard to visualize the connection to the ambient $d$-dimensional space of the data points (i.e. for smallest enclosing ball we could use a $(d+1)$-variate $\eps$-quantization to measure the $d$ coordinates of the center point and the radius).
Instead, for a summarizing shape we may wish to study a \emph{shape inclusion probability function} $h_{\mu_P} : \b{R}^d \to [0,1]$ (or \sip\ function) which describes the probability that a given point $x \in \b{R}^d$ is included in the summarizing shape\footnote{For technical reasons, if there are (degenerately) multiple optimal summarizing shapes, we say each are equally likely to be the summarizing shape of the point set.}.
Again, there does not seem to be a closed form for many of these functions.
Rather we calculate an $\eps$-\sip\ function $\hat{h} : \b{R}^d \to [0,1]$ such that
$
\forall_{x \in \b{R}^d} \left| h(x) - \hat{h}(x) \right| \leq \eps.
$
The size of an $\eps$-sip depends only on $\eps$ and the complexity of the summarizing shape.

\subsection{Our Results}
We describe simple and practical randomized algorithms for computing $\eps$-quantizations, $\eps$-\sip\ functions, and $(\eps, \alpha)$-kernels.  Let $T_f(n)$ be the time it takes to calculate a summarizing shape of a set of $n$ points $Q \subset \b{R}^d$, which generates a statistic $f(Q)$.  We can calculate an $\eps$-quantization of $f^{\leq}_{\mu_P}$, with probability $1-\delta$, in time $O(T_f(n) \frac{1}{\eps^2} \log \frac{1}{\eps \delta})$.  For univariate $\eps$-quantizations the size is $O(\frac{1}{\eps})$, and for $k$-variate $\eps$-quantizations the size is $O(k^2 \frac{1}{\eps} \log^{2k} \frac{1}{\eps} )$.
With probability $1-\delta$, we can calculate an $\eps$-\sip\ function of size $O(\frac{1}{\eps^2} \log \frac{1}{\eps \delta})$ in time $O(T_f(n) \frac{1}{\eps^2} \log \frac{1}{\eps \delta})$.
With probability $1-\delta$, we can calculate an $(\eps, \alpha)$-kernel of size $O(\frac{1}{\alpha^{(d-1)/2}} \frac{1}{\eps^2} \log \frac{1}{\eps \delta})$ in time $O((n + \frac{1}{\alpha^{d-3/2}}) \frac{1}{\eps^2} \log \frac{1}{\eps \delta})$.
All of these randomized algorithms are simple and practical, as demonstrated by some experimental results.

In addition, we provide deterministic algorithms for computing $\eps$-quantizations of a specific class of functions. If $\c A$ is a family of geometric shapes, such that $(\b{R}^d, \c{A})$ has bounded VC-dimension, and $f : \b R^{dn} \to \b R^k$ is a function that describes some statistics on the smallest element from $\c A$ that encloses the points (e.g. the radius of the smallest enclosing ball), then an $\eps$-quantization for $f$ can be computed in deterministic time $O( \textrm{poly}(n,\frac{1}{\eps}))$, as described in Table \ref{tbl:results}.

This paper describes results for shape fitting problems for distributions of point sets in $\b{R}^d$, in particular, we will use the smallest enclosing ball and the axis-aligned bounding box as running examples in the algorithm descriptions. We believe, though, that the concept of $\eps$-quantizations should extend to many other problems with uncertain data. In fact, variations of our randomized algorithm should work for a more general array of problems.

\section{Preliminaries: $\eps$-Samples and $\alpha$-Kernels}
\label{sec:prelim}

\paragraph{$\eps$-Samples.}
For a set $P$ (in our context a point set), let $\c{A}$ be a set of subsets of $P$ which for instance could be induced by containment in a shape from some family of geometric shapes.  
The pair $(P, \c{A})$ is called a \emph{range space}.  We say that $Q$ is an \emph{$\eps$-sample} of $(P,\c{A})$ if 
$$
\forall_{R \in \c{A}} \left|\frac{\phi(R \cap Q)}{\phi(Q)} - \frac{\phi(R \cap P)}{\phi(P)}\right| \leq \eps,
$$
where $|\cdot|$ takes the absolute value and $\phi(\cdot)$ returns the measure of a point set.  In the discrete case $\phi(Q)$ returns the cardinality of $Q$.  We say $\c{A}$ \emph{shatters} a set $S$ if every subset of $P$ is equal to $R \cap S$ for some $R \in \c{A}$.  The cardinality of the largest discrete set $X \subseteq P$ that $\c{A}$ can shatter is known as the \emph{VC-dimension} of $(P,\c{A})$.  

When $(P, \c{A})$ has constant VC-dimension, we can create an $\eps$-sample $Q$ of $(P, \c{A})$, with probability $1-\delta$, by uniformly sampling $O(\nu \frac{1}{\eps^2} \log \frac{\nu}{\delta \eps})$ points from $P$~\cite{VC71}.  There exist deterministic techniques to create $\eps$-samples~\cite{Mat91,CM96} of size $O(\nu \frac{1}{\eps^2} \log \frac{1}{\eps})$ in time $O(\nu^{3\nu} n (\frac{1}{\eps^2} \log \frac{\nu}{\eps})^\nu)$.  
When $P$ is a point set in $\b{R}^d$ and the family of ranges $\c{Q}_k$ is determined by inclusion of convex shapes whose sides have one of $k$ predefined normal directions, such as the set of axis-aligned boxes, then 
an $\eps$-sample for $(P, \c{Q}_k)$ of size $O(\frac{k}{\eps} \log^{2k} \frac{1}{\eps})$ can be constructed in $O(\frac{n}{\eps^3} \log^{6k} \frac{1}{\eps})$ time~\cite{Phi08}.  

When we have a distribution $\mu : \b{R}^d \to \b{R}^+$, such that $\int_{x \in \b{R}} \mu(x) \; dx = 1$, we can think of this as the set $P$ of points in $\b{R}^d$, where the weight $w$ of a point $p \in \b{R}^d$ is $\mu(p)$.  
To simplify notation, we write $(\mu, \c{A})$ as a range space where the ground set is this set $P = \b{R}^d$ weighted by the distribution $\mu$.  Let it have VC-dimension $\nu$.  
For distribution $\mu$ that is polygonally approximable~\cite{Phi08} with a constant number of facets, we can construct an $\eps$-sample of size $O(\frac{\nu}{\eps^2} \log \frac{\nu}{\eps})$ in time $O(\frac{\nu}{\eps^2} \log^2 \frac{\nu}{\eps})$.
A longer primer on $\eps$-samples is in Appendix \ref{app:eps-sample}.

\paragraph{$\alpha$-Kernels.}
Given a point set $P \in \b{R}^d$ of size $n$ and a direction $u \in \b{S}^{d-1}$, let $P[u] = \arg\max_{p \in P} \IP{p}{u}$, where $\IP{\cdot}{\cdot}$ is the inner product operator.  
Let $\wid(P,u) = \IP{P[u] - P[-u]}{u}$ describe the width of $P$ in direction $u$.  
We say that $K \subseteq P$ is an \emph{$\alpha$-kernel} of $P$ if for all $u \in \b{S}^{d-1}$
$$
\wid(P,u) - \wid(K,u) \leq \alpha \cdot \wid(P,u).
$$
$\alpha$-kernels of size $O(\frac{1}{\alpha^{(d-1)/2}})$ can be calculated in time $O(n + \frac{1}{\alpha^{d-3/2}})$~\cite{Cha06}.  Computing many extent related problems such as diameter and smallest enclosing ball on the $\alpha$-kernel approximates the function on the original set~\cite{AHV04}.

\section{Randomized Algorithm for $\eps$-Quantizations}

We start with a general algorithm (Algorithm \ref{alg:rand-draw}) which will be made specific in several places in the paper.  
We assume we can draw a point from $\mu_p$ for each $p \in P$ in constant time; if the time depends on some other parameters, the runtimes can be easily adjusted.  

\begin{algorithm}[h!!t]
\caption{Approximate $\mu_P$ with regard to a family of shapes $\c{S}$ or function $f_{\c{S}}$
\label{alg:rand-draw}}
\begin{algorithmic}[1]
\FOR {$i = 1$ \textbf{to} $m = O(\frac{1}{\eps^2} \log \frac{1}{\eps \delta})$}
  \FOR {$p_j \in P$}
    \STATE Generate $q_j \in \mu_{p_j}$.
  \ENDFOR
  \STATE Set $V_i = f_{\c{S}}(\{q_1, q_2, \ldots, q_n\})$.  
\ENDFOR
\STATE Reduce or Simplify the set $V = \{ V_i\}_{i=1}^m$.
\end{algorithmic}
\end{algorithm}

\paragraph{Defining $\eps$-quantizations.}
For an isotonic function $h : \b{R} \to [0,1]$, an $\eps$-quantization, $R$, is a set of points where for any $t \in \b{R}$, $| h(t) - R(t) | \leq \eps$.  We let $R(t) = \frac{1}{|R|} \sum_{r \in R} 1(r \leq t)$.  Since $h$ has range $[0,1]$ and is isotonic, an $\eps$-quantization requires only $O(1/\eps)$ points.  
Figure \ref{fig:1e-quants} shows a illustration of how an $\eps$-quantization approximates a smooth function.  
Because $h$ is isotonic there exists a function $g: \b{R} \to \b{R}^+$ such that $h(t) = \int_{x=-\infty}^t g(x) \; dx$ where $\int_{x=-\infty}^\infty g(x) \; dx = 1$.  
Thus an $\eps$-sample of $(g, \c{I}_+)$ is an $\eps$-quantization of $h$, where $\c{I}^+$ is all $1$-sided intervals.   

\vierplaatjes {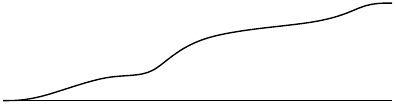} {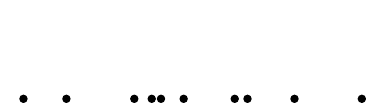} {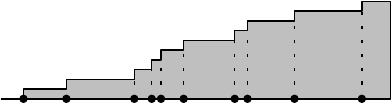} {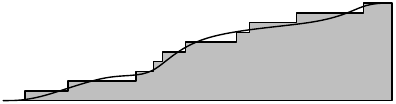}
{\label{fig:1e-quants}
  (a) The true form of the function. (b) The $\eps$-quantization as a point set in $\b{R}$. (b) The inferred curve in $\b{R}^2$. (d) Overlay of the two images.}

For an isotonic function $h : \b{R}^k \to [0,1]$ a $k$-variate $\eps$-quantization, $R$, is a set of points in $\b{R}^k$ such that for any $p \in \b{R}^k$, $| h(p) - R(p) | \leq \eps$.  
For $p \in \b{R}^k$ let $R(p) = \frac{1}{|R|} \sum_{q \in R} 1(q \preceq p)$.  
Because $h$ is isotonic, there exists a function $g : \b{R}^k \to \b{R}^+$ such that $h(p) = \int_{x \preceq p} g(x) \; dx$ and $\int_{x \in \b{R}^d} g(x) \; dx = 1$.
Thus an $\eps$-sample of $(g, \c{R}_+)$ is an $\eps$-quantization of $h$, where $\c{R}_+$ describes ranges $R_p \in \c{R}_+$ defined by all $q$ such that $q \preceq p$ for any $p$.    See Figure \ref{fig:ke-quants} for an illustration of $k$-variate function $h$ and a $k$-variate $\eps$-quantization approximating it.

\vierplaatjes {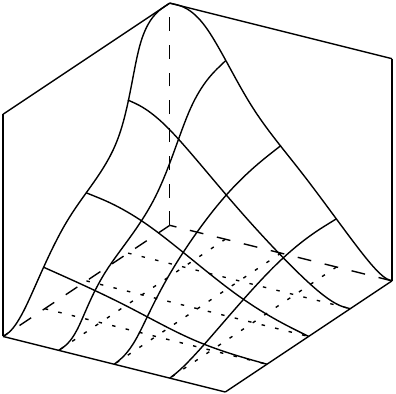} {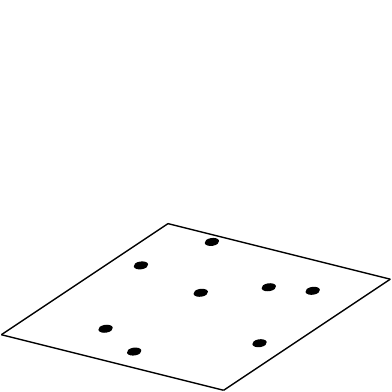} {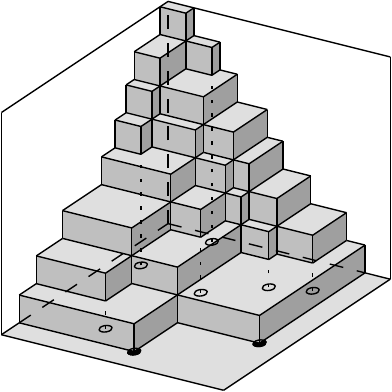} {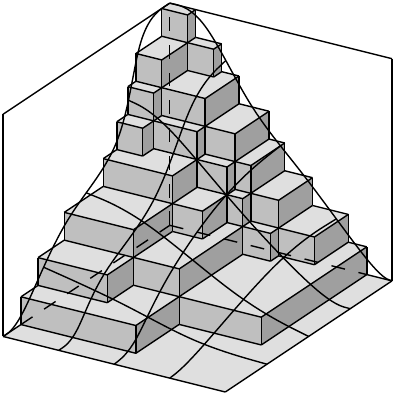}
{\label{fig:ke-quants}
  (a) The true form of the multivariate function. (b) The $\varepsilon$-quantisation as a point set in $k$-space. (b) The inferred surface in $k+1$-space. (d) Overlay of the two images.}

\paragraph{Algorithm for $\eps$-quantizations.}
For a function $f$ on a point set $P$ of size $n$, it takes $T_f(n)$ time to evaluate $f(P)$.
We now construct $f^{\leq}_{\mu_P}$ by adapting Algorithm \ref{alg:rand-draw} as follows.  First draw a sample point $q_j$ from each $\mu_{p_j}$ for $p_j \in P$, then evaluate $V_i = f(\{q_1, \ldots, q_n\})$.  The fraction of trials of this process that produces a value less than $v$ is the estimate of $f^{\leq}_{\mu_P}(v)$.  
Finally reduce the size of $V$ be returning $\frac{2}{\eps}$ evenly spaced points according to the sorted order.

\begin{theorem}
For a distribution $\mu_P$ of $n$ points, 
there exists a univariate $\eps$-quantization of size $O(\frac{1}{\eps})$ for $f^{\leq}_{\mu_P}$, and it can be constructed in $O(T_f(n) \frac{1}{\eps^2} \log \frac{1}{\eps \delta})$ time, with success probability $1-\delta$, where $T_f(n)$ is the time it takes to compute $f(Q)$ for any point set $Q$ of size $n$.  
\end{theorem}
\begin{proof}
Because $f^{\leq}_{\mu_P} : \b{R} \to [0,1]$ is an isotonic function, there exists another function $g: \b{R} \to \b{R}^+$ such that $f^{\leq}_{\mu_P}(t) = \int_{x = -\infty}^t g(x) \; dx$ where $\int_{\b{R}} g(x) \;dx = 1$.  And thus an $\eps$-sample of $(g, \c{I}_+)$ is an $\eps$-quantization of $f^{\leq}_{\mu_P}$.  

By drawing a random sample $q_i$ from each $\mu_{p_i}$ for $p_i \in P$, we are drawing a random point set $Q$ from $\mu_P$.  Thus $f(Q)$ is a random sample from $g$.  Hence, using the standard randomized construction for $\eps$-samples, $O(\frac{1}{\eps^2} \log \frac{1}{\eps \delta})$ such samples will generate an $\frac{\eps}{2}$-sample for $g$, and hence an $\frac{\eps}{2}$-quantization for $f^{\leq}_{\mu_P}$, with probability $1-\delta$.  

Since in an $\frac{\eps}{2}$-quantization, every value is off from the true function by at most $\frac{\eps}{2}$, then we can take an $\frac{\eps}{2}$-quantization of the step function and still have an $\eps$-quantization of the true function.  
Thus, we can reduce this to an $\eps$-quantization of size $O(\frac{1}{\eps})$ by taking a subset of $\frac{2}{\eps}$ points spaced evenly according to their sorted order.  
\end{proof}

We can construct $k$-variate $\eps$-quantizations using the same basic procedure as in Algorithm \ref{alg:rand-draw}.  The output $V_i$ of $f_{\c{S}}$ is $k$-variate and thus results in a $k$-dimensional point.  As a result, the reduction of the final size of the point set requires more advanced procedures.  

\begin{theorem}
For a distribution $\mu_P$ of $n$ points, 
there exists a $k$-variate $\eps$-quantization of size $O(\frac{k^2}{\eps} \log^{2k} \frac{1}{\eps})$ for $f^{\preceq}_{\mu_P}$, and it can be constructed in $O(T_f(n) \frac{k}{\eps^2} \log \frac{k}{\eps \delta} + k^2 \frac{1}{\eps^5} \log^{6k} \frac{1}{\eps} \log \frac{1}{\eps \delta})$ time, with success probability $1-\delta$, where $T_f(n)$ is the time it takes to compute $f(Q)$ for any point set $Q$ of size $n$.  
\label{thm:k-var-q}
\end{theorem}
\begin{proof}
In the $k$-variate case there exists a function $g : \b{R}^k \to \b{R}^+$ such that $f^{\preceq}_{\mu_P}(v) = \int_{x \preceq v} g(x) \; dx$ where $\int_{\b{R}^k} g(x) \; dx = 1$.  Then a random point set $Q$ from $\mu_P$, evaluated as $f(Q)$, is still a random sample from the $k$-variate distribution described by $g$.  Thus, with probability $1-\delta$, a set of $O(\frac{k}{\eps^2} \log \frac{1}{\eps \delta})$ such samples is an $\eps$-sample of $(g,\c{R}_+)$, which has VC-dimension $k$, and the samples are also a $k$-variate $\eps$-quantization of $f^{\preceq}_{\mu_P}$.  

We can then reduce the size of the $\eps$-quantization to $O(\frac{k}{\eps} \log^{2k} \frac{1}{\eps})$ \cite{Phi08} (or to $O(\frac{k}{\eps^2} \log \frac{1}{\eps})$ \cite{CM96}), since the VC-dimension is $k$ and each data point requires $O(k)$ storage.  
\end{proof}

\section{$(\eps, \alpha)$-Kernels}
The above construction works for a fixed family of summarizing shapes.  This section builds a single data structure, an $(\eps, \alpha)$-kernel, for a distribution $\mu_P$ in $\b{R}^d$ that can be used to construct $(\eps, \alpha)$-quantizations for several families of summarizing shapes.  In particular, an $(\eps, \alpha)$-kernel of $\mu_P$ is a data structure such that in any query direction $u \in \b{S}^{d-1}$ we can create an $(\eps, \alpha)$-quantization of $\wid(\cdot, u)$, the width in direction $u$.  
This data structure introduces a parameter $\alpha$, which deals with geometric error, in addition to the error parameter $\eps$, which deals with probability error. 

We follow the randomized framework described above as follows.
Let $\c{K}$ be an $(\eps, \alpha)$-kernel consisting of $m = O(\frac{1}{\eps^2} \log \frac{1}{\eps \delta})$ $\alpha$-kernels, where each $\alpha$-kernel $K_j$ approximates a point set $Q_j$ drawn randomly from $\mu_P$.  Given $\c{K}$, we can then create an $(\eps, \alpha)$-quantization for the width of $\mu_P$ in any direction $u \in \b{S}^{d-1}$.  
Specifically, let $M = \{\wid(K_j, u) \}_{j=1}^m$.

\begin{lemma}
With probability $1-\delta$, 
$M$ is an $(\eps, \alpha)$-quantization of the width of $\mu_P$ in direction $u$.
\end{lemma}
\begin{proof}
The width $\wid(Q_j, u)$ of a random point set $Q_j$ drawn from $\mu_P$ is a random sample from the distribution over widths of $\mu_P$ in direction $u$.  Thus, with probability $1-\delta$, $m$ such random samples would create an $\eps$-quantization.  Using the width of the $\alpha$-kernels $K_j$ instead of $Q_j$ induces an error on each random sample of at most $\alpha \cdot \wid(Q_j, u)$.  
Then for a query width $w$, say there are $\gamma m$ point sets $Q_j$ that have width $\leq w$ and $\gamma^\prime m$ $\alpha$-kernels $K_j$ with width $\leq w$.  Note that $\gamma^\prime>\gamma$.
Let $\hat{w} = w - \alpha w$. For each point set $Q_j$ that has width $> w$ but the corresponding $\alpha$-kernel $K_j$ has width $\leq w$, it follows that $K_j$ has width $> \hat{w}$.  
Thus the number of $\alpha$-kernels $K_j$ that have width $\leq \hat{w}$ is $\leq \gamma m$, and thus there is a width $w^\prime$ between $w$ and $\hat{w}$ such that the number of $\alpha$-kernels $\leq w^\prime$ is exactly $\gamma m$.  
\end{proof}

\begin{theorem}
With probability $1-\delta$, we can construct an $(\eps, \alpha)$-kernel for $\mu_P$ on $n$ points in $\b{R}^d$ of size $O(\frac{1}{\alpha^{(d-1)/2}} \frac{1}{\eps^2} \log \frac{1}{\eps \delta})$ and in time $O((n + \frac{1}{\alpha^{d-3/2}}) \frac{1}{\eps^2} \log \frac{1}{\eps \delta})$.  
\end{theorem}

\paragraph{$k$-Dependent $(\eps, \alpha)$-Kernels.}
The definition of $(\eps, \alpha)$-quantizations can be extended to $k$-variate $(\eps, \alpha)$-quantizations $M$ where 
(1) there exists a point $x^\prime \in \b{R}^k$ such that for all integers $i \in [1,k]$ $|x^{(i)} - (x^\prime)^{(i)}| \leq \alpha x^{(i)}$ and
(2) $|M(x) - f_{\mu_P}^{\preceq}(x^\prime)| \leq \eps$.
Let $x^{(i)}$ represent the $i$th coordinate of a point $x \in \b{R}^k$.

$(\eps,\alpha)$-kernels can be generalized to approximate other functions $f : \b{R}^{dn} \to \b{R}^k$, specified as follows.
We say a point $p^\prime \in \b{R}^k$ is a \emph{relative $\theta$-approximation} of $p \in \b{R}^k$ if for each coordinate $i$ we have $p^{(i)} - {p^\prime}^{(i)} \leq \theta p^{(i)}$.  For functions $f$ and $\theta$ where $f(K)$ is a relative $\theta(\alpha)$-approximation of $f(Q)$ when $K$ is an $\alpha$-kernel of $Q$, we say that $f$ is \emph{relative $\theta(\alpha)$-approximable}.  

By setting $m = O(\frac{k}{\eps^2} \log \frac{k}{\eps \delta})$ in the above algorithm, with probability $1-\delta$, we can build a \emph{$k$-dependent $(\eps, \alpha)$-kernel} data structure $\c{K}$ with the following properties.  
It has size $O(\frac{1}{\alpha^{(d-1)/2}} \frac{k}{\eps^2} \log \frac{k}{\eps \delta})$ and can be built in time $O((n + \frac{1}{\alpha^{d-3/2}})  \frac{k}{\eps^2} \log \frac{k}{\eps \delta})$.  
To create a $k$-variate $(\eps, \alpha)$-quantization for a function $f$, create a $k$-dimensional point $p_j = f(K_j)$ for each $\alpha$-kernel $K_j$ in $\c{K}$.  The set $M$ of $m$ $k$-dimensional points forms the $k$-variate $(\eps, \alpha)$-quantization.

\begin{theorem}
Let $f$ be a relative $\theta(\alpha)$-approximable function that takes $T_f(n)$ time to evaluate on a set of $n$ points in $\b{R}^d$.  
From a $k$-dependent $(\eps,\alpha)$-kernel $\c{K}$ with $m$ $\alpha$-kernels, with probability $1-\delta$, we can create a $k$-variate $(\eps, \theta(\alpha))$-quantization of $f$, of size $O(\frac{1}{\eps} \log^{2k} \frac{1}{\eps})$ in time $O(T_f(\frac{1}{\alpha^{(d-1)/2}}) m)$.
\end{theorem}
\begin{proof}
Each evaluation of $f$ on a point set $Q_j$ drawn from $\mu_P$ is a random sample from the distribution over $f$ on point sets drawn from $\mu_P$ and hence these values on all $m$ sampled point sets would be an $\eps$-quantization of $f^{\preceq}_{\mu_P}$.  

For a query point $w \in \b{R}^k$, let $\gamma m$ point sets produce a value $w_j = f(Q_j)$ such that $w_j \preceq w$, and let $\gamma^\prime m$ point sets produce a value $w_j^\prime = f(K_j)$ such that $w_j^\prime \preceq w$.  Note that $\gamma^\prime > \gamma$.  
Because $f$ is relative $\theta(\alpha)$-approximable, for each point set $Q_j$ such that $w_j^\prime \preceq w$, but $w_j \npreceq w$, then $w_j^\prime \npreceq \hat{w}$, where $\hat{w} = w - \theta(\alpha) w$.  (More specifically, for each coordinate $w^{(i)}$ of $w$, $\hat{w}^{(i)} = w^{(i)} - \theta(\alpha) w^{(i)}$.)  
Thus, the number of point sets such that $f(K_j) \preceq \hat{w}$ is $\leq \gamma m$, and hence there is a point $w^\prime$ between $w$ and $\hat{w}$ such that the fraction of sampled point sets such that $f(K_j) \preceq w^\prime$ is exactly $\gamma$, and hence is within $\eps$ of the true fraction of point sets sampled from $\mu_P$ with probability $1-\delta$.
\end{proof}

To name a new examples, the width and diameter are relative $\alpha$-approximable functions, thus the results apply directly with $k=1$.  The radius of the minimum enclosing ball is relative $2\alpha$-approximable with $k=1$.
The $d$ directional widths of the minimum perimeter or minimum volume axis-aligned rectangle is relative $\alpha$-approximable with $k=d$.

\subsection{Experiments with $(\eps, \alpha)$-Kernels and $\eps$-Quantizations}

We implemented these randomized algorithms for $(\eps, \alpha)$-kernels and $\eps$-quantizations for diameter (\sf{diam}), width in a fixed direction (\sf{dwid}), and radius of the smallest enclosing $\ell_2$ ball (\sf{seb}$_2$).  We used existing code from Hai Yu~\cite{YAPV04} for $\alpha$-kernels and Bernd G\"{a}rtner~\cite{Gar99} for $\sf{seb}_2$.  For the input set $\mu_P$ we generated $5000$ points $P \subset \b{R}^3$ on the surface of a cylinder piece with radius $1$ and axis length $10$.  Each point $p \in P$ represented the center of a Gaussian with standard deviation $3$.  We set $\eps = .2$ and generated $\alpha$-kernels of size at most $40$ (the existing code did not allow the use to specify a parameter $\alpha$, only the maximum size).  
We generated a total of $m=40$ point sets from $\mu_P$.
The $(\eps, \alpha)$-kernel has a total of $1338$ points.
We calculated $\eps$-quantizations and $(\eps, \alpha)$-quantizations for \sf{diam}, \sf{dwid}, and \sf{seb}$_2$, each of size $10$; see Figure \ref{fig:exp}.

\drieplaatjes {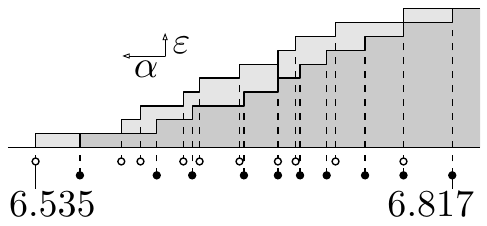} {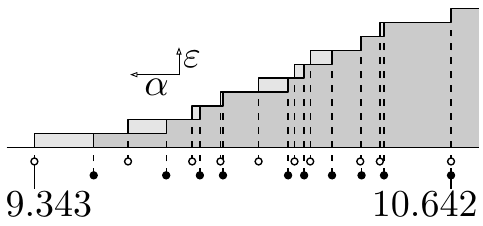} {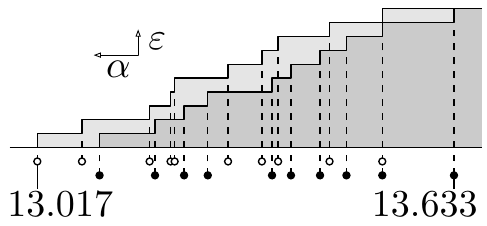}
  {\label{fig:exp} $(\eps,\alpha)$-quantization (white circles) and $\eps$-quantization (black circles) for (a) \sf{seb}$_2$, (b) \sf{dwid}, and (c) \sf{diam}.}







\section{Shape Inclusion Probabilities}
We can also use a variation of Algorithm \ref{alg:rand-draw} to construct $\eps$-shape inclusion probability functions.  For a point set $Q \subset \b{R}^d$, let the summarizing shape $S_Q = \c{S}(Q)$ be from some geometric family $\c{S}$ so $(\b{R}^d, \c{S})$ has bounded VC-dimension $\nu$.  
We randomly sample point sets $Q_j$ from $\mu_P$ and then find the summarizing shape $S_{Q_j}$ (e.g. minimum enclosing ball) of $Q_j$.  Let this set of shapes be $S^{(\mu_P)}$.  If there are multiple shapes from $\c{S}$ which are equally optimal (as can happen degenerately with, for example, minimum width slabs), choose one of these shapes at random.  
For a set of shapes $S^\prime \subset \c{S}$, let $S^\prime_p \subset S^\prime$ be the subset of shapes that contain $p \in \b{R}^d$.  
We store $S^{(\mu_P)}$ and evaluate a query point $p \in \b{R}^d$ by counting what fraction of the shapes the point is contained in, specifically returning $|S^{(\mu_P)}_p| / |S^{(\mu_P)}|$ in $O(\nu |S^{(\mu_P)}|)$ time.  In some cases, this evaluation can be sped up with point location data structures.

\begin{theorem}
For a distribution $\mu_P$ of $n$ points and a family of summarizing shapes $(\b{R}^d, \c{S})$ with bounded VC-dimension $\nu$, with probability $1-\delta$
we can construct an $\eps$-\sip function of size $O(2^{\nu+1} \frac{\nu^2}{\eps^2} \log \frac{1}{\eps \delta})$ 
and in time $O(T_{\c{S}}(n)\frac{1}{\eps^2} \log \frac{1}{\eps \delta})$, where $T_{\c{S}}(n)$ is the time it takes to determine the summarizing shape of any point set $Q \subset \b{R}^d$ of size $n$.
\label{thm:rand-sip}
\end{theorem}
\begin{proof}
If $(\b{R}^d, \c{S})$ has VC-dimension $\nu$, then the dual range space $(\c{S}, P^*)$ has VC-dimension $\nu^\prime \leq 2^{\nu+1}$, where $P^*$ is all subsets $\c{S}_p \subseteq \c{S}$, for any $p \in \b{R}^d$, such that $\c{S}_p = \{S \in \c{S} \mid p \in S\}$.  
Using the above algorithm, sample $m = O(\frac{\nu^\prime}{\eps^2} \log \frac{\nu^\prime}{\eps \delta})$ point sets $Q$ from $\mu_P$ and generate the $m$ summarizing shapes $S_Q$.  
Each shape is a random sample from $\c{S}$ according to $\mu_P$, and thus $S^{(\mu_P)}$ is an $\eps$-sample of $(\c{S}, P^*)$.

Let $w_{\mu_P}(S)$, for $S \in \c{S}$, be the probability that $S$ is the summarizing shape of a point set $Q$ drawn randomly from $\mu_P$.  
Let $W_{\mu_P}(\c{S}^\prime) = \int_{S \in \c{S}^\prime} w_{\mu_P}(S)$, where $\c{S}^\prime \subseteq P^*$, be the probability that some shape from the subset $\c{S}^\prime$ is the summarizing shape of $Q$ drawn from $\mu_P$.  

We approximate the \sip function at $p \in \b{R}^d$ by returning the fraction $|S^{(\mu_P)}_p| / m$.  
The true answer to the \sip function at $p \in \b{R}^d$ is $W_{\mu_P}(\c{S}_p)$.  
Since $S^{(\mu_P)}$ is an $\eps$-sample of $(\c{S}, P^*)$, then with probability $1-\delta$
$$
\left| \frac{|S^{(\mu_P)}_p|}{m} - \frac{W_{\mu_P}(\c{S}_p)}{1} \right|
= 
\left| \frac{|S^{(\mu_P)}_p|}{|S^{(\mu_P)}|} - \frac{W_{\mu_P}(\c{S}_p)}{W_{\mu_P}(P^*)} \right|
\leq \eps.
$$

Since for the family of summarizing shapes $\c{S}$ the range space $(\b{R}^d, \c{S})$ has VC-dimension $\nu$, each can be stored using that much space.  
\end{proof}

The size can then be reduced to $O(2^{\nu +1} \frac{\nu^2}{\eps^2} \log \frac{1}{\eps})$ in time $O((2^{\nu+1})^{3 \cdot 2^{\nu+1}+1} (\frac{\nu}{\eps} \log \frac{1}{\eps})^{2^{\nu+1}+1})$ using deterministic techniques.

\paragraph{Representing $\eps$-\sip functions by Isolines.}
  \vierplaatjes {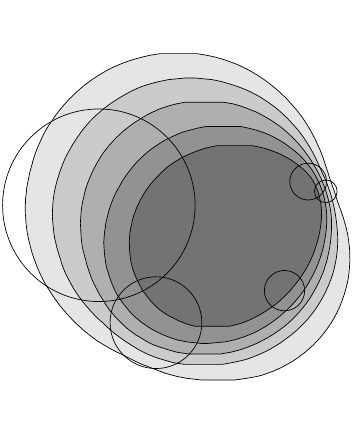} {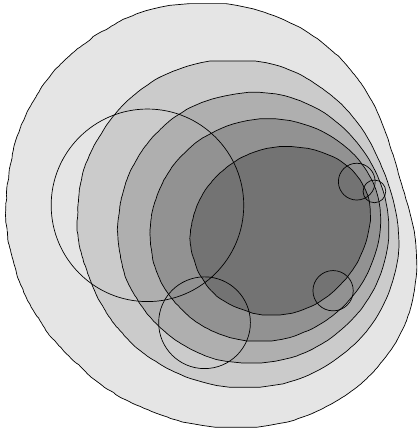} {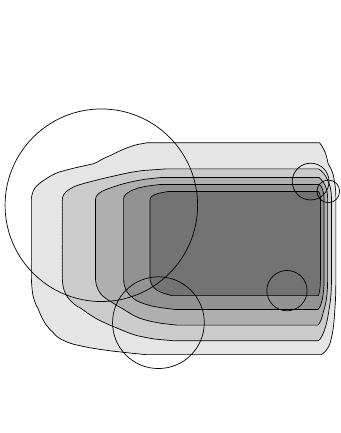} {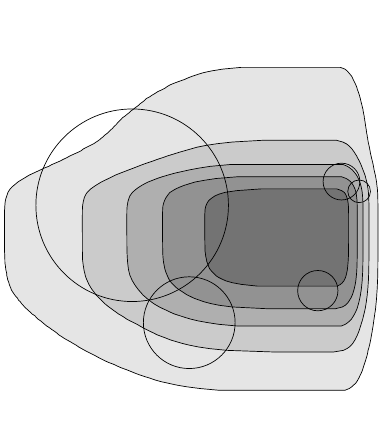}
  {\label{fig:sip-seb} (a) The shape inclusion probability for the smallest enclosing ball, for points uniformly distributed inside the circles. 
  (b) The same, but for normally distributed points around the circle centers, with standard deviations given by the radii.  
  (c) The shape inclusion probability for the smallest enclosing axis-aligned rectangle, for points uniformly distributed inside the circles. 
  (d) The same, but for normally distributed points.}

Shape inclusion probability functions are density functions.  One convenient way of visually representing a density function in $\b{R}^2$ is by drawing the isolines.  A \emph{$\gamma$-isoline} is a closed curve such that on the inside the density function is $> \gamma$ and on the outside is $< \gamma$.  

In each part of Figure \ref{fig:sip-seb} a set of 5 circles correspond to points with a probability distribution.   For part (a) and (c), the probability distribution is uniform over those circles, in part (b) and (d) it is drawn from a multivariate Gaussian distribution with standard deviation as the radius.  We generate $\eps$-\sip functions for smallest enclosing ball in Figure \ref{fig:sip-seb}(a,b) and for smallest axis-aligned bounding box in Figure \ref{fig:sip-seb}(c,d).  

In all figures we draw approximations of $\{.9,.7,.5,.3,.1\}$-isolines. 
These drawing are generated by randomly selecting $m = 5000$ (a,b) or $m=25000$ (c,d) shapes, counting the number of inclusions at different points in the plane and interpolating to get the isolines.
  The innermost and darkest region has probability $> 90\%$, the next one probability $> 70\%$, etc., the outermost region has probability $< 10 \%$.

When $\mu_P$ describes the distribution for $n$ points and $n$ is large, then isolines are generally connected for convex summarizing shapes.  In fact, in $O(n)$ time we can create a point which is contained in the convex hull of a point set sampled from $\mu_P$ with high probability.  Specifics are discussed in Appendix \ref{app:center-point}.

\section{Deterministic Constructions of $\eps$-Quantizations}
\label{sec:deterministic}

In this section we consider functions $f$ which describe the size of some summarizing shape from the family $\c{A}$ such that $(\b{R}^d, \c{A})$ has constant VC-dimension.
In particular, given a point set $Q \subset \b{R}^d$, let $\c{A}(Q) \subset \b{R}^d$ (e.g. smallest enclosing ball) be the summarizing shape for $Q$, and let $f(Q)$ be a statistic of $\c{A}(Q)$ (e.g. radius of the smallest enclosing ball).
The overall strategy will be to deterministically approximate each $\mu_{p_i}$ with a point set $Q_{p_i}$, although not with respect to the range space $(\mu_{p_i}, \c{A})$, but with a more complicated range space described below.  Let $Q_P = \{Q_{p_i}\}_i$ describe this set of point sets.  Then let the function 
$f(Q_P, r)$ describe the fraction of point sets $Q^\prime = (q_1 \in Q_{p_1}, q_2 \in Q_{p_2}, \ldots, q_n \in Q_{p_n})$ for $\{Q_{p_1}, \ldots Q_{p_n}\} = Q_P$ such that $f(Q^\prime) \leq r$.  
We show that we can generate a set of point sets $Q_P$ such that $f(Q_P, r)$ is a good approximation of $f^{\leq}_{\mu_P}(r)$.  
And we show how to efficiently evaluate $f(Q_P,r)$.

\subsection{Approximating $\mu_p$}
In this section we restrict that $\mu_P$ is either defined by a polygonal surface $S$ with $b$ facets or is polygonal approximable, it can be approximated by a finite polygonal surface $S$ with $b$ facets, for some constant $b$, as in \cite{Phi08}.  

It might seem that we can just create an $\eps$-sample of $(\mu_{p_i}, \c{A})$ for each $\mu_{p_i}$, but we need to consider a more complicated family $\c{A}_{f,n}$.  Given a family of shapes $\c{A}$ and a function $f$ which computes a value determined by a summarizing shape $A \in \c{A}$ for a set of $n$ points, then $\c{A}_{f,n}$ is a family of shapes where each is defined by a set of $n-1$ points $T \subset \b{R}^d$ and a value $w$.  Specifically, $\c{A}_{f,n}(T,w)$ is the set of points $\{p \in \b{R}^d \mid f(T \cup p) \leq w \}$.  

In certain cases, such as the volume of the axis-aligned bounding box, $(\mu_{p_i}, \c{A}_{f,n})$ has constant VC-dimension.  Shapes from $\c{A}_{f,n}$ are determined by the placement of $2d$ points, the most extreme in each axis direction, thus its shatter dimension is $\sigma_f = 2d$.  Hence an $\eps$-sample for $(\mu_{p_i}, \c{A}_{f,n})$ of size $O(\frac{1}{\eps^2} \log \frac{1}{\eps})$ can be calculated in time $O(\frac{1}{\eps^2} \log^2 \frac{1}{\eps})$ for each $\mu_{p_i}$.  

\parpic[r]{\PPic{2cm}{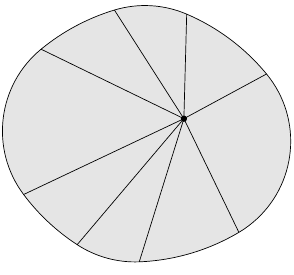}}

In other cases, such as the radius of smallest enclosing $\ell_2$ disks, $\c{A}_{f,n}$ defines regions which have $O(n)$ $(d-1)$-dimensional faces on its boundary and thus $(\mu_{p_i}, \c{A}_{f,n})$ has VC-dimension $n$.  
Naive techniques would take time exponential in $n$ to deterministically create an $\eps$-sample, but we can do better by decomposing a shape $A \in \c{A}_{f,n}$ into $O(n)$ disjoint simpler shapes.
In the case of disks, $\c{A}_{f,n}$ has its boundary defined by at most $2n$ circular arcs of two different radii.  We can choose a point in the convex hull of $T$ and draw lines to each intersection of circular arcs, see the figure on right.  The intersections of the disc defining each boundary piece and the halfspaces for the drawn lines at its endpoints describes a \emph{wedge} from a family $\c{W}_{f,n}$.  The range space $(\b{R}^d, \c{W}_{f,n})$ has VC-dimension at most $9$ because shapes from $\c{W}_{f,n}$ are formed by the intersection of three shapes from families that would each have VC-dimenion $3$ in a range space on the same ground set.  
Thus, an $\frac{\eps}{2n}$-sample of $(\mu_{p_i}, \c{W}_{f,n})$ is an $\eps$-sample of $(\mu_{p_i}, \c{A}_{f,n})$.  
So for radius of the smallest enclosing $\ell_2$ balls we can create an $\eps$-sample of $(\mu_{p_i}, \c{A}_{f,n})$ of size $O(n^2 \frac{1}{\eps^2} \log \frac{n}{\eps})$ in time $O(n^2 \frac{1}{\eps^2} \log^2 \frac{n}{\eps})$ for each $\mu_{p_i}$.

We generalize both of these cases to other shapes and in higher dimensions in Appendix \ref{app:shapes}.  There are also illustrations of various shapes from $\c{A}_{f,n}$.

\begin{lemma}
When each $\mu_{p_i}$ is approximated with an $\eps^\prime$-sample $Q_{p_i}$ of $(\mu_{p_i},\c{A}_{f,n})$, then for any $r$
$$
\left| \P[f_{\mu_P}(P) \leq r]  -
f(\{Q_{p_1}, Q_{p_2}, \ldots, Q_{p_n}\}, r)  \right| \leq  \eps^\prime n.
$$
\label{thm:n-apx}
\end{lemma}
\begin{proof}
When $P$ is drawn from a distribution $\mu_P$, then we can write that probability that $f_{\mu_P}(P) \leq r$ as follows.

\begin{eqnarray*}
\P[f_{\mu_P}(P) \leq r] = \int_{q_1} \mu_{p_1}(q_1) \int_{q_2} \mu_{p_2}(q_2) \ldots \int_{q_n} \mu_{p_n}(q_n) 
1(f(\{q_1, q_2, \ldots, q_n\}) \leq r) \;
d q_n d q_{n-1} \ldots d {q_1}
\end{eqnarray*}

Consider the inner most integral
$$
\int_{q_n} \mu_{p_n}(q_n)  1(f(\{q_1, q_2, \ldots, q_n\}) \leq r) \; d q_n,
$$
where $\{q_1, q_2 \ldots, q_{n-1}\}$ are fixed.  The indicator function is true when for $q_n$ $f(\{q_1, q_2, \ldots, q_{n-1}, q_n\}) \leq r$ and hence $q_n$ is contained in a shape from $\c{A}_{f,n}(\{q_1, q_2, \ldots q_{n-1}\}, r)$.  
Thus if we have an $\eps^\prime$-sample $Q_{p_n}$ for $(\mu_{p_n},\c{A}_{f,n})$, then we can guarantee that 
$$
\int_{q_n}  \mu_{p_n}(q_n)  1(f(\{q_1, q_2, \ldots, q_n\}) \leq r) \; d q_n \leq \frac{1}{|Q_{p_n}|} \sum_{q_n \in Q_{p_n}} 1(f(\{q_1, q_2, \ldots, q_{n-1}, q_n\}) \leq r) + \eps^\prime.
$$
We can then move the $\eps^\prime$ to the outside, and we can change the order of the integrals to write:
\begin{eqnarray*}
\lefteqn{\P[f_{\mu_P}(P) \leq r] \; \leq }
\\ & & \frac{1}{|Q_{p_n}|} \sum_{q_n \in Q_{p_n}} \int_{q_1} \mu_{p_1}(q_1) \int_{q_2} \mu_{p_2}(q_2) \ldots \int_{q_{n-1}} \mu_{p_{n-1}}(q_{n-1}) 
1(f(\{q_1, q_2, \ldots, q_n\}) \leq r) \;
d q_{n-1} d q_{n-2} \ldots d {q_1} + \eps^\prime.
\end{eqnarray*}
Repeating this procedure $n$ times we get:
\begin{eqnarray*}
\P[f_{\mu_P}(P) \leq r] 
&\leq &
\left(\prod_{i=1}^n \frac{1}{|Q_{p_i}|}\right) 
\sum_{i=1}^n \sum_{q_i \in Q_{p_i}} 
1(f(\{q_1, q_2, \ldots, q_n\}) \leq r)
+ \eps^\prime n.
\\ & = &
f(Q_P,r) + \eps^\prime n.
\end{eqnarray*}

Using the same technique we can achieve a symmetric lower bound for $\P[f_{\mu_P}(P) \leq r]$.
\end{proof}

By setting $\eps^\prime = \eps/n$ we can achieve an additive $\eps$-approximation by using an $\eps^\prime$-sample for each $(\mu_{p_i}, \c{A}_{f,n})$.

\subsection{Evaluating $f(Q_P,r)$.}
Evaluating $f(Q_P, r)$ in time polynomial in $n$ and $|Q_{p_i}|$, for any $i$, is not completely trivial since there are $n^{|Q_{p_i}|}$ possible sets in $Q_P$.  
Let a \emph{good set} be a set of $n$ points, $G$, such that for each $Q_i$ there exists a point $g_i \in G$ such that $g_i \in Q_i$.  For each good set $G$ there exists a unique basis of at most $\sigma_f$ points\footnote{This uniqueness requires careful construction of the $\eps$-samples $Q_i$, as described in Appendix \ref{sec:eps-dist}.} which define the summarizing shape of $G$ (remember the shatter dimension of $\c{A}$ is $\sigma_f$ and $\sigma_f < \nu_f$, the VC-dimension).
Define a \emph{valid basis} to be a set of at most $\sigma_f$ points in $Q_P$ such that each point is from a different $Q_i$ and if any point is removed the summarizing shape changes.
Each valid basis forms a basis for several good sets.  

We now construct $R$, an $\eps$-quantization of $f^{\leq}_{\mu_P}$.
This approximation is created by calculating the summarizing shape for all good sets.  Even though there are an exponential number of good sets, there are only a polynomial number of valid bases.  Thus for each valid basis, we count the number of good sets it represents.  And we let each valid basis contribute to the $\eps$-quantization; its position is determined by its value in $f$ and its weight by the number of good sets it represents.  We initially store the $\eps$-quantization as a sorted list of tuples $(r,\eta)$ where $r = f(\{q_1, q_2, \ldots, q_{\sigma_f}\})$ for some valid basis $\{q_1, q_2, \ldots, q_{\sigma_f}\}$, and $\eta$ is the fraction of the good sets which are represented by this valid basis.  The details are outlined in Algorithm \ref{alg:count-prob}.

\begin{algorithm}[h!!t]
\caption{\label{alg:count-prob}Construct $\eps$-Quantization from $Q_P$}
\begin{algorithmic}[1]
\FOR {all valid bases $q_1, q_2, \ldots, q_{\sigma_f} \in Q_P$}
  \FOR {$i = 1$ \textbf{to} $n$}
    \IF {$q_1 \in Q_i$ or $q_2 \in Q_i$ or $\ldots$ or $q_{\sigma_f} \in Q_i$}
      \STATE Set $w_i = \frac{1}{|Q_i|}$.
    \ELSE
      \STATE Set $w_i = \frac{1}{|Q_i|} \sum_{q_j \in Q_i} 1(q_j \in \c{A}(\{q_1, q_2, \ldots, q_{\sigma_f}\}))$
    \ENDIF    
  \ENDFOR
  \STATE Insert $(f(q_1, q_2, \ldots, q_{\sigma_f}), \prod_i w_i)$ into $R$.
\ENDFOR
\end{algorithmic}
\end{algorithm}

We now summarize the full deterministic algorithm.
For each $(\mu_{p_i}, \c{A}_{f.n})$ we create an $\frac{\eps}{n}$-sample $Q_{p_i}$ of size $\alpha_{f}(n,\eps)$.
This makes the set $Q_P$ have $\eta = \sum_{i=1}^n |Q_{p_i}| = n \alpha_f(n,\eps)$ points in its sets.  We examine $O(\eta^{\sigma_f})$ valid bases.
For each valid basis we can evaluate $f(G)$  and compute $w_i$ in $\sf{RS}_f(n,\eps)$ time using a range searching data structure, after preprocessing or with a naive search.
Thus the deterministic running time for constructing an $\eps$-quantization is $O(\eta^{\nu_f} \sf{RS}_f(n,\eps))$ which is presented for various summarizing shapes in Table \ref{tbl:Afn-size}.  
For instance, for volume of the axis-aligned bounding box this takes $O(n^{6d}/\eps^{4d} \log^{3d} \frac{n}{\eps})$ time and for radius of the smallest enclosing disks this takes $O(n^{16.5}/\eps^{7} \log^{3.5} \frac{n}{\eps})$.
The total construction time for the $\eps$-quantizations is the sum of this time and the time to construct $n$ $(\eps/n)$-samples of $(\b{R}^d, \c{A}_{f,n})$; for both smallest enclosing disks and for axis-aligned bounding boxes it is the former.

A univariate $\eps$-quantization can be reduced to size $O(\frac{1}{\eps})$.  
Furthermore, we can create $k$-variate $\eps$-quantizations using the same procedure (such as the width in the $k$ dimensions of an axis-aligned bounding box).  The condition in Lemma \ref{thm:n-apx} where $f_{\mu_P}(P) \leq r$ can be replaced with a $k$-variate condition $f_{\mu_P}(P) \preceq r$ for $r \in \b{R}^k$.  Thus the same argument applies when we define $f : \b{R}^{dn} \to \b{R}^k$, and we can create $k$-variate $\eps$-quantizations of size $k^2 \frac{1}{\eps} \log^{O(k)} \frac{k}{\eps}$ in the same deterministic times as long as $\nu_f = O(k)$.  

\begin{theorem}
For any range space $(\mu_P, \c{A}_f)$ for a distribution $\mu_P$ of $n$ points, 
  with VC-dimension $\nu_f$, 
  where each $(\mu_{p_i}, \c{A}_{f,n})$ has an $\frac{\eps}{n}$-sample of size $\alpha_f(n, \eps)$,
  and where, after preprocessing $m$ points and with near-linear space and time, we can count the number of points in a shape from $\c{A}_f$ in $\textsf{RS}(m,\c{A}_f)$ time, 
we construct a $k$-variate $\eps$-quantization of $f^{\preceq}_{\mu_P}$ of size $k^2 \frac{1}{\eps} \log^{O(k)} \frac{k}{\eps}$ in $O((n \alpha_f(n, \eps))^{\nu_f} \cdot \textsf{RS}((n \alpha_f(n, \eps))^{\nu_f},\c{A}_f))$ time.
\end{theorem}

\section{Acknowledgements}
We would like to thank Pankaj K. Agarwal for many helpful discussions and Sariel Har-Peled for suggesting the use of wedges.

\bibliographystyle{plain}
\bibliography{uncert}

\appendix

\section{Primer on $\eps$-Samples}
\label{app:eps-sample}

We recall from Section \ref{sec:prelim} that for a range space $(P, \c{A})$ an $\eps$-sample $Q \subseteq P$ guarantees 
$$
\forall_{R \in \c{A}} \left|\frac{\phi(R \cap Q)}{\phi(Q)} - \frac{\phi(R \cap P)}{\phi(P)}\right| \leq \eps,
$$
where $|\cdot|$ takes the absolute value and $\phi(\cdot)$ returns the measure of a point set.  In the discrete case $\phi(Q)$ returns the cardinality of $Q$.  

When $P \subset \b{R}^d$ we describe a few common examples of $\c{A}$. 
Let $\c{B}$ describe all subsets of $P$ determined by containment in some ball.
Let $\c{R}_d$ describe all subsets of $P$ defined by containment in some $d$-dimensional axis-aligned box.
Let $\c{H}$ describe all subsets of $P$ defined by containment in some halfspace.  
Throughout the paper we use $\c{A}$ generically to represent one such family of ranges.  

Also recall from Section \ref{sec:prelim} that if $(P, \c{A})$ has bounded VC-dimension $\nu$, then we can create an $\eps$-sample, with probability $1-\delta$, by sampling $O(\frac{\nu}{\eps^2} \log \frac{\nu}{\eps \delta})$ points at random, or deterministically of size $O(\frac{\nu}{\eps^2} \log \frac{\nu}{\eps})$ in time $O(\nu^{2 \nu} n (\frac{1}{\eps^2} \log \frac{\nu}{\eps})^\nu)$.  There exist $\eps$-samples of slightly smaller sizes~\cite{MWW93}, but efficient constructions are not known.  
If $(P, \c{A})$ has VC-dimension $\nu$, this also implies that $(P, \c{A})$ contains at most $|P|^\nu$ sets.  

Similarly, the \emph{shatter function} $\pi_{(P, \c{A})}(m)$ of a range space $(P, \c{A})$ is the maximum number of sets $S \in (P,\c{A})$ where $|S|=m$.  The \emph{shatter dimension} $\sigma$ of a range space $(P, \c{A})$ is the minimum value such that $\pi_{(P,\c{A})}(m) = O(m^{\sigma})$.  It can be shown~\cite{HPbook} that $\sigma \leq \nu$ and $\nu = O(\sigma \log \sigma)$.  

For a range space $(P, \c{A})$ the \emph{dual range space} is defined $(\c{A}, P^*)$ where $P^*$ is all subsets $\c{A}_p \subseteq \c{A}$ defined for an element $p \in P$ such that $\c{A}_p = \{A \in \c{A} \mid p \in A\}$.  If $(P,\c{A})$ has VC-dimension $\nu$, then $(\c{A}, P^*)$ has VC-dimension $\leq 2^{\nu+1}$.  Thus, if the VC-dimension of $(\c{A}, P^*)$ is constant, then the VC-dimension of $(P,\c{A})$ is also constant \cite{Mat99}.  Hence, the standard $\eps$-sample theorems apply to dual range spaces as well.  

Let $g : \b{R} \to \b{R}^+$ be a function where $\int_{x=-\infty}^\infty g(x) \; dx = 1$.  
We can create an $\eps$-sample $Q_g$ of $(g,\c{I}_+)$, where $\c{I}_+$ describes the set of all one-sided intervals of the form $(-\infty, t)$, so that
$$
\max_{t} \left|\int_{x = -\infty}^{t} g(x) \; dx - \frac{1}{|Q_g|} \sum_{q \in Q_g} 1(q < t) \right| \leq \eps.
$$
We can construct $Q_g$ of size $O(\frac{1}{\eps})$ by choosing a set of points in  $Q_g$ so that the integral between two consecutive points is always $\eps$.  But we do not need to be so precise.  Consider the set of $\frac{2}{\eps}$ points $\{ q^\prime_1, q^\prime_2, \ldots, q^\prime_{\frac{2}{\eps}} \}$ such that $\int_{x=-\infty}^{q_i^\prime} = i \eps /2$.   Any set of $\frac{2}{\eps}$ points $Q_g = \{q_1, q_2, \ldots, q_{\frac{2}{\eps}} \}$ such that $q^\prime_i \leq q_i \leq q^\prime_{i+1}$ is an $\eps$-sample.

\subsection{$\eps$-Samples of Distributions.}
\label{sec:eps-dist}
We say a subset $W \subset \b{R}^d$ is \emph{polygonal approximable} if there exists a polygonal shape $S$ with $m$ facets such that $\phi(W \setminus S) + \phi(S \setminus W) \leq \eps \phi(W)$ for any $\eps > 0$.  Usually, $m$ is dependent on $\eps$.  In turn, such a polygonal shape $S$ describes a continuous point set where $(S, \c{A})$ can be given an $\eps$-sample $Q$ using $O(\frac{1}{\eps^2} \log \frac{1}{\eps})$ points if $(S, \c{A})$ has bounded VC-dimension \cite{Mat99} or using $O(\frac{1}{\eps} \log^{2k} \frac{1}{\eps})$ points if $\c{A}$ is defined by a constant $k$ number of directions \cite{Phi08}.  For instance, where $\c{A} = \c{B}$ is the set of all balls then the first case applies, and when $\c{A} = \c{R}_2$ is the set of all axis-aligned rectangles then either case applies.  

A shape $W \subset \b{R}^{d+1}$ may describe a distribution $\mu : \b{R}^d \to [0,1]$.  We note that many common distributions like multivariate Gaussian distributions are polygonally approximable.  For instance for a range space $(\mu, \c{B})$, then the range space of the associated shape $W_\mu$ is $(W_\mu, \c{B} \times \b{R})$ where $\c{B} \times \b{R}$ describes balls in $\b{R}^d$ for the first $d$ coordinates and any points in the $(d+1)$th coordinate.  

The general scheme to create an $\eps$-sample for $(S, \c{A})$, where $S \in \b{R}^d$ is a polygonal shape, is to use a lattice $\Lambda$ of points.   
A \emph{lattice} $\Lambda$ in $\b{R}^d$ is an infinite set of points defined such that for $d$ vectors $\{v_1, \ldots, v_d\}$ that form a basis, for any point $p \in \Lambda$, $p + v_i$ and $p - v_i$ are also in $\Lambda$ for any $i \in [1,d]$.  
We first create a discrete $\frac{\eps}{2}$-sample $M \subset \Lambda$ of  $(S, \c{A})$ and then create an $\frac{\eps}{2}$-sample $Q$ of $(M, \c{A})$ using standard techniques~\cite{CM96,Phi08}.  Then $Q$ is an $\eps$-sample of $(S, \c{A})$.
For a shape $S$ with $m$ $(d-1)$-faces on its boundary, any subset $A^\prime \subset \b{R}^d$ that is described by a subset from $(S, \c{A})$ is an intersection $A^\prime = A \cap S$ for some $A \in \c{A}$.  Since $S$ has $m$ $(d-1)$-dimensional faces, we can bound the VC-dimension of $(S, \c{A})$ as $\nu = O((m + \nu_{\c{A}}) \log (m+\nu_{\c{A}}))$ where $\nu_{\c{A}}$ is the VC-dimension of $(\b{R}^d, \c{A})$.  
Finally the set $M = S \cap \Lambda$ is determined by choosing an arbitrary initial origin point in $\Lambda$ and then uniformly scaling all vectors $\{v_1, \ldots, v_d\}$ until $|M| = \Theta(\frac{\nu}{\eps^2} \log \frac{\nu}{\eps})$~\cite{Mat99}.  This construction follows a less general but smaller construction in Phillips~\cite{Phi08}.  

It follows that we can create such an $\eps$-sample of size $|M|$ in time $O(|M| m \log |M|)$ by starting with a scaling of the lattice so a constant number of points are in $S$ and then doubling the scale until we get to within a factor of $d$ of $|M|$.  If there are $n$ points inside $S$, it takes $O(nm)$ time to count them. 

\begin{lemma}
For a polygonal shape $S \subset \b{R}^d$ with $m$ facets, we can construct an $\eps$-sample for $(S, \c{A})$ of size $O(\frac{\nu}{\eps^2} \log \frac{\nu}{\eps})$ in time $O(m \frac{\nu}{\eps^2} \log^2 \frac{\nu}{\eps})$, where $(S, \c{A})$ has VC-dimension $\nu_{\c{A}}$ and $\nu = O((\nu_{\c{A}}+m) \log (\nu_{\c{A}}+m))$.  
\label{lem:mu-eps}
\end{lemma}

An important part of the above construction is the arbitrary choice of the origin points of the lattice $\Lambda$.  This allows us to arbitrarily shift the lattice defining $M$ and thus the set $Q$.  In Section \ref{sec:deterministic} we need to construct $n$ $\eps$-samples $\{Q_1, \ldots, Q_n\}$ for $n$ range spaces $\{(S_1, \c{A}), \ldots, (S_n, \c{A})\}$.  
In Algorithm \ref{alg:count-prob} we examine sets of $\nu_{\c{A}}$ points, each from separate $\eps$-samples that define a minimal shape $A \in \c{A}$.  It is important that we do not have two such (possibly not disjoint) sets of $\nu_{\c{A}}$ points that define the same minimal shape $A \in \c{A}$.  (Note, this does not include cases where say two points are antipodal on a disk and any other point in the disk added to a set of $\nu_{\c{A}}=3$ points forms such a set; it refers to cases where say four points lie (degenerately) on the boundary of a disc.)
We can guarantee this by enforcing a property on all pairs of origin points $p$ and $q$ for $(S_i, \c{A})$ and $(S_j, \c{A})$.  For the purpose of construction, it is easiest to consider only the $l$th coordinates $p_l$ and $q_l$ for any pair of origin points or lattice vectors (where the same lattice vectors are used for each lattice).  We enforce a specific property on every such pair $p_l$ and $q_l$, for all $l$ and all distributions and lattice vectors.  

First, consider the case where $\c{A} = \c{R}_d$ describes axis-aligned bounding boxes.  It is easy to see that if for all pairs $p_l$ and $q_l$ that $(p_l - q_l)$ is irrational, then we cannot have $>2d$ points on the boundary of an axis-aligned bounding box, hence the desired property is satisfied.  

Now consider the more complicated case where $\c{A} = \c{B}$ describes smallest enclosing balls.  There is a polynomial of degree $2$ that describes the boundary of the ball, so we can enforce that for all pairs $p_l$ and $q_l$ that $(p_l - q_l)$ is of the form $c_1 (r_{p_l})^{1/3} + c_2 (r_{q_l})^{1/3}$ where $c_1$ and $c_2$ are rational coefficients and $r_{p_l}$ and $r_{q_l}$ are distinct integers that are not multiple of cubes.    Now if $\nu = d+1$ such points satisfy (and in fact define) the equation of the boundary of a ball, then no $(d+2)$th point which has this property with respect to the first $d+1$ can also satisfy this equation.  

More generally, if $\c{A}$ can be described with a polynomial of degree $p$ with $\nu$ variables, then enforce that every pair of coordinates are the sum of $(p+1)$-roots.  This ensures that no $\nu+1$ points can satisfy the equation, and the undesired situation cannot occur.

\section{A Center Point for $\mu_P$}
\label{app:center-point}
We can create a point $\bar{q} \in \b{R}^d$ that is in the convex hull of a sampled point set $Q$ from $\mu_P$ with high probability.  This implies that for any summarizing shape that contains the convex hull, $\bar{q}$ is also contained in that summarizing shape.  Let $\c{H}$ be the family of subsets defined by halfspaces.
We use the following algorithm:
\begin{enumerate}
\item Create $2$-approximate center points $\bar{p}_i$ for each $\mu_{p_i}$ (i.e. using a $(1/4)$-sample of $(\mu_{p_i}, \c{H})$).  Let the set be $\bar{P}$.
\item Create $2$-approximate center point $\bar{q}$ of $\bar{P}$.  
\end{enumerate}
All steps can be done in $O(n)$ time because we can create $(1/4)$-samples of all range spaces $(\mu_{p_i}, \c{H})$ and of $(\bar{P},\c{H})$ in $O(n)$ time.  Constructing approximate center points can be done in $O(1)$ time on a constant sized set, such as $(1/4)$-sample~\cite{CEMST96}.  

\begin{lemma}
Given a distribution of a point set $\mu_P$ (such that each point distribution is polygonally approximable) of $n$ points in $\b{R}^d$, there is an $O(n)$ time algorithm to create a point $\bar{q}$ that will be in the convex hull of a point set drawn from $\mu_P$ with probability $\geq 1 - ((1-1/(2d+2))^{1/(2d+2)})^n$.  
\label{lem:ext-center}
\end{lemma}
\begin{proof}
For each $\bar{p}_i \in \bar{P}$, any halfspace that has $\bar{p}_i$ on its boundary and does not contain $\bar{q}$ has probability $\geq 1/(2d+2)$ of containing a random point from $\mu_{p_i}$.  
Thus for any direction $u \in \b{S}^{d-1}$ there are at least $n/(2d+2)$ points $\bar{p}_i$ from $\bar{P}$ for which $\IP{q}{u} \leq \IP{\bar{p}_i}{u}$.  And for each of those points $\bar{p}_i$, the probability that the point $q_i$ sampled from $\mu_{p_i}$ is such that $\IP{\bar{p}_i}{u} \leq \IP{q_i}{u}$ (and thus $\IP{\bar{q}}{u} \leq \IP{q_i}{u}$) is $\leq 1/(2d+2)$. 
Hence, the probability that there is a separating halfspace between $\bar{q}$ and the convex hull of $Q$ (where the halfspace is orthogonal to some direction $u$) is $\leq (1-1/(2d+2))^{n/(2d+2)} = ((1-1/(2d+2))^{1/(2d+2)})^n$.   
\end{proof}

\begin{theorem}
For a set of $m < n$ point sets drawn i.i.d. from $\mu_P$, it follows that $\bar{q}$ is in each of the $m$ convex hulls for each point sets with high probability (specifically with probability $\geq 1 - m\left((1-1/(2d+2))^{1/(2d+2)}\right)^n$).
\end{theorem}
\begin{proof}
Let $\beta = (1-1/(2d+2))^{1/(2d+2)}$.  
For any one point set the probability that $\bar{q}$ is contained in the convex hull is $> 1-\beta^n$.  By the union bound, the probability that it is contained in all $m$ convex hulls is $> (1-\beta^n)^m = 1 - m \beta^n + {m \choose 2} \beta^{2n} - {m \choose 3} \beta^{3n} + \ldots$.  
Since $n > m$, the sum of all terms after the first two in the expansion increase the probability.
\end{proof}

Thus because the summarizing shapes are convex, then for any point $q$, the line segment $\overline{q \bar{q}}$ is completely contained in a convex summarizing shape if and only if $q$ is.  Thus for every boundary of a summarizing shape $\overline{q \bar{q}}$ crosses, $q$ is outside that summarizing shape.  This implies the following corollary.

\begin{corollary}
If the summarizing shape is convex, then the $\gamma$-layer, for $\gamma < 1-1/m$, exists, is connected, and is star-shaped with high probability, specifically with probability $\geq 1 - m\left((1-1/(2d+2))^{1/(2d+2)}\right)^n$.  
\end{corollary}

\section{Shapes of $\c{A}_{f,n}$ for Various Summarizing Shapes}
\label{app:shapes}
Let $\c{A}_{f,n}$ be the intersection of $O(n)$ shapes from $\hat{\c{A}}_f$ where $(\mu_p, \hat{\c{A}}_f)$ has VC-dimension $\hat{\nu}_f$.  
Let a \emph{wedge} of $\c{A}_{f,n}$ be a shape from $\c{W}_{f,n}$ described by the intersection of $d$ hyperplanes and one shape from $\hat{\c{A}}_f$.  

\begin{lemma}
The VC-dimension of $(\b{R}^d, \c{W}_{f,n})$ is $d(d+1) + \hat{\nu}_f$.  
\end{lemma}
\begin{proof}
It is known that a class of shapes $\c{W}_k$ that is formed as the intersection of $k$ subset from $(\b{R}^d, \c{A}_j)$ for $j \in [1:k]$ which have VC-dimension $\nu_j$, then $(\b{R}^d, \c{W}_k)$ has VC-dimension $\sum_{j=1}^k \nu_j$~\cite{HPbook}.  
Since wedges are formed by the intersection of $d$ halfspaces (VC-dimension $d+1$) and one shape from $\hat{\c{A}}_f$, it follows that $(\b{R}^d, \c{W}_{f,n})$ has VC-dimension $d(d+1) + \hat{\nu}_f$.  
\end{proof}

If the $(d-2)$-dimensional faces of the boundary of $\c{A}_{f,n}$ are subsets of $(d-2)$-dimensional flats (i.e. points in $\b{R}^2$ and line segments in $\b{R}^3$), then any shape from $\c{A}_{f,n}$ can be formed as the disjoint union of $O(n)$ wedges from $\c{W}_{f,n}$.  
Functions which produce such families $\c{A}_{f,n}$ include $\sf{seb}_2$ and \sf{chp} in $\b{R}^2$ and \sf{diam} and \sf{cha} in $\b{R}^d$.  

\begin{remark}\emph{
In cases, such as $\sf{seb}_2$ and \sf{chp} for $d>2$, where we cannot form wedges, we can create similar shapes for $\c{W}_{f,n}$, as generalized cones whose boundary passes through the boundary of each $(d-1)$-dimensional facet of the corresponding shape from $\c{A}_{f,n}$.  For these shapes the VC-dimension of $(\b{R}^d, \c{W}_{f,n})$ can be bounded as $O(\nu_f d \log d)$.   
Each face of a generalized cone is described by two shapes from $\c{A}$, which have VC-dimension $\nu_f$, and a point.  Thus the face of the generalized cone has shatter dimension $O(\nu_f)$.  If a $(d-1)$-dimensional facet of the boundary of a shape from $\c{A}_{f,n}$ has more than $d$ faces of dimension $(d-2)$, then we can triangulate the facet so it has $O(d)$ such faces.  Thus the range space for the generalized cone has shatter dimension $O(\nu_f d)$ and VC-dimension $O(\nu_f d \log d)$.  
}\end{remark}

The VC-dimension for $(\b{R}^d, \c{W}_{f,n})$, $\psi_f$, is shown for several functions in Table \ref{tbl:Afn-VC}.  

\begin{lemma}
If the disjoint union of $m$ shapes from $\c{W}_{f,n}$ can form any shape from $\c{A}_{f,n}$, then  
an $\frac{\eps}{m}$-sample of $(\mu_p, \c{W}_{f,n})$ is an $\eps$-sample of $(\mu_p, \c{A}_{f,n})$.
\end{lemma}
\begin{proof}
For any shape $A \in \c{A}_{f,n}$ we can create a set of $m$ shapes $\{W_1, \ldots, W_n\} \subset \c{W}_{f,n}$ whose disjoint union is $A$.  Since each range of $\c{W}_{f,n}$ may have error $\frac{\eps}{m}$, their union has error at most $\eps$.  
\end{proof}

Hence for $(\mu_p, \c{W}_{f,n})$  $\frac{\eps}{n}$-samples can be created of size $O(n^2 \frac{1}{\eps^2} \log \frac{n}{\eps})$ in time $O(n (\frac{n}{\eps})^{2 \psi_f} \log^{\psi_f} \frac{n}{\eps})$. 

\begin{table}[ht]
\small
\caption{Runtimes for $\eps$-Quantizations of Various Summarizing Shape Families.}
\centering
\begin{tabular}{|l|l|c|c|c|}
\hline 
abbrv. & summarizing shape & randomized$^*$& determ. $\b{R}^2$ & determ. $\b{R}^d$
\\ \hline \hline
\sf{dwid} & width along a fixed direction & $O(n \frac{1}{\eps^2} \log \frac{1}{\eps})$ & $O( n^{4}/\eps)$ & $O(n^{4}/\eps)$
\\ \hline
\sf{aabbp} & axis-aligned bounding box measured by perimeter & $O(n \frac{1}{\eps^2} \log \frac{1}{\eps})$ & $\tilde{O}(n^8 / \eps^4)$ & $\tilde{O}(n^{6d} / \eps^{4d})$
\\ \hline
\sf{aabba} & axis-aligned bounding box measured by area & $O(n \frac{1}{\eps^2} \log \frac{1}{\eps})$ & $\tilde{O}(n^{12}/\eps^8)$ & $\tilde{O}(n^{6d} / \eps^{4d})$
\\ \hline
\sf{seb$_\infty$} & smallest enclosing ball, $L_\infty$ metric & $O(n \frac{1}{\eps^2} \log \frac{1}{\eps})$ & $\tilde{O}(n^6/\eps^3)$ & $\tilde{O}(n^{2d+2} / \eps^{d+1})$
\\ \hline
\sf{seb$_1$} & smallest enclosing ball, $L_1$ metric & $O(n \frac{1}{\eps^2} \log \frac{1}{\eps})$ & $\tilde{O}(n^6/\eps^3)$ & $\tilde{O}(n^{2d+2} / \eps^{d+1})$
\\ \hline
\sf{seb$_2$} & smallest enclosing ball, $L_2$ metric & $O(n \frac{1}{\eps^2} \log \frac{1}{\eps})$ & $\tilde{O}(n^{16.5}/\eps^{7})$ & $\tilde{O}((n^{d+3}/\eps^2)^{d+2 - 1/d})$
\\ \hline
\sf{diam} & diameter & $O(n^2 \frac{1}{\eps^2} \log \frac{1}{\eps})$ & $\tilde{O}((n^{5}/\eps^{2})^{n+1})$ & $\tilde{O}((n^5 /\eps^2)^{n+1})$
\\ \hline
\sf{cha} & convex hull measured by area & $O(n \log n \frac{1}{\eps^2} \log \frac{1}{\eps})$ & $\tilde{O}((n^5/\eps^2)^{n+1})$ & $\tilde{O}((n^{d+3}/\eps^2)^{n+1})$
\\ \hline
\sf{chp} & convex hull measured by perimeter & $O(n \log n \frac{1}{\eps^2} \log \frac{1}{\eps})$ & $\tilde{O}((n^5/\eps^2)^{n+1})$ & $\tilde{O}((n^{d+3}/\eps^2)^{n+1})$
\\ \hline
\end{tabular}
\\ $^*$ all randomized results are correct with constant probability.
\\ $\tilde{O}(f(n,\eps))$ ignores poly-logarithmic factors $(\log \frac{n}{\eps})^{O(\textrm{poly}(d))}$ ,  \;\;\;\; for any $\tau>0$.
\label{tbl:results}
\end{table}

\subsection{Examples}
In the examples below, 7 points are given, on which we study a certain measure (e.g., diameter or convex hull area). The grey region denotes the possible placements of a new point, such that the measure will not exceed a given value.  These regions illustrate $\c{A}_{f,n}$ for various summarizing shapes.

  \begin{table}[ht]
\caption{VC-dimension for Various Shape Families.}
\small
\centering
\begin{tabular}{|l|c|c|c|c|}
\hline 
abbrv. &  $(\b{R}^2, \c{A}_{f,n})$ & $(\b{R}^d, \c{A}_{f,n})$ & $(\b{R}^2, \c{W}_{f,n})$& $(\b{R}^d, \c{W}_{f,n})$
\\ \hline \hline
\sf{aabbp} & $O(1)$ $(\sigma = 4)$ & $O(d \log d)$ $(\sigma = 2d)$ & &
\\ \hline
\sf{aabba} &  $8$ & $O(d \log d)$ $(\sigma = 2d)$ & &
\\ \hline
\sf{seb$_\infty$} & $4$ & $2d$ & &
\\ \hline
\sf{seb$_1$} &  $4$ & $2d$ & &
\\ \hline
\sf{seb$_2$} &  $\infty$ & $\infty$ & $9$ & $O(d^2 \log d)$
\\ \hline
\sf{diam} &  $\infty$ & $\infty$ & $9$ & $d^2 + 2d +1$
\\ \hline
\sf{cha} &  $\infty$ & $\infty$ & $7$ & $d^2 + 2d +1$
\\ \hline
\sf{chp} &  $\infty$ & $\infty$ & $O(1)$ & $O(d^2\log d)$
\\ \hline
\end{tabular}
\label{tbl:Afn-VC}
\end{table}

\paragraph{Axis-aligned bounding box.}
Figure \ref{fig:aaba} shows examples of $\c{A}_{f,n}$ for axis-aligned bounding boxes, measuring either by perimeter (\sf{aabbp}) or by area (\sf{aabba}) in $\b{R}^2$.  For both $(\b{R}^2, \c{A}_{f,n})$ has a shatter dimension of $4$ because the shape is determined by the $x$-coordinates of $2$ points and the $y$-coordinates of $2$ points.  This generalizes to a shatter dimension of $2d$ for $(\b{R}^d, \c{A}_{f,n})$, where area generalizes to $d$-dimensional volume, and perimeter generalizes to the $(d-1)$-volume of the boundary.  We can also show the VC-dimension of $(\b{R}^2, \c{A}_{f,n})$ is $8$ for \sf{aabbp} because its shape is defined by the intersection of halfspaces with $4$ predefined normal directions at $0^\circ$, $45^\circ$, $90^\circ$, and $135^\circ$.  This can be generalized to higher dimensions.  

  \tweeplaatjes {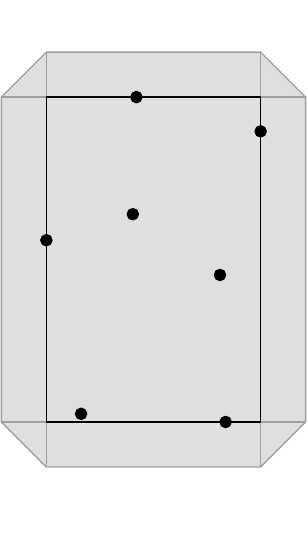} {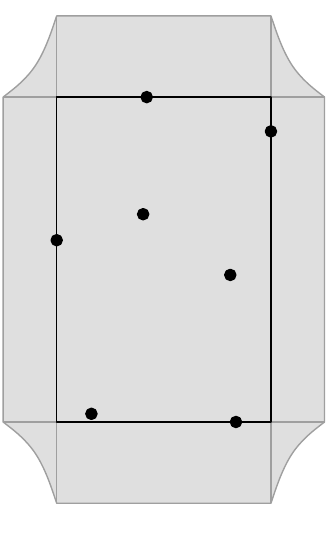}
  {\label{fig:aaba} (a) Axis-aligned bounding box, measured by perimeter. (b) Axis-aligned bounding box, measured by area. The curves are hyperbola parts.}

Hence, for both shapes we can create $n$ $\frac{\eps}{n}$-samples of $(\mu_{p_i}, \c{A}_{f,n})$ of size $\alpha_f(n,\eps) = O(\frac{n^2}{\eps^2} \log \frac{n}{\eps})$ in time $O(\frac{n^3}{\eps^2} \log^2 \frac{n}{\eps})$.  For \sf{aabbp} in $\b{R}^2$, an $\frac{\eps}{n}$-sample of each $(\mu_{p_i}, \c{A}_{f,n})$ of can be reduced further to size $O(\frac{n}{\eps} \log^{16} \frac{n}{\eps})$ in total time $O(\frac{n^5}{\eps^4} \log^{40} \frac{n}{\eps})$.  
Then we can construct the $\eps$-quantization in $(n^{6d} / \eps^{4d}) (\log \frac{n}{\eps})^{O(d)}$ time, using orthogonal range searching.  For \sf{aabbp} in $\b{R}^2$, the runtime improves to $O(n^8/\eps^4 \log^{65} \frac{n}{\eps})$.

\paragraph{Smallest enclosing ball.}
Figure \ref{fig:seb} shows examples of $\c{A}_{f,n}$ for smallest enclosing ball, for metrics $L_\infty$ (\sf{seb}$_{\infty}$), $L_1$ (\sf{seb}$_1$), and $L_2$ (\sf{seb}$_2$) in $\b{R}^2$.  For $\sf{seb}_\infty$ and $\sf{seb}_1$, $(\b{R}^d, \c{A}_{f,n})$ has VC-dimension $2d$ because the shapes are defined by the intersection of halfspaces from $d$ predefined normal directions.  
For $\sf{seb}_1$ and $\sf{seb}_\infty$, we can create $n$ $\frac{\eps}{n}$-samples of each $(\mu_{p_i}, \c{A}_{f,n})$ of size $\alpha_f(n, \eps) = O(\frac{n^2}{\eps^2} \log \frac{n}{\eps})$ in total time $O(\frac{n^3}{\eps^2} \log^2 \frac{n}{\eps})$.  The size for each can be reduced to $O(\frac{n}{\eps} \log^{2d} \frac{n}{\eps})$ in $O(\frac{n^5}{\eps^4} \log^{8d} \frac{n}{\eps})$ total time.  Using an orthogonal range searching data structure we can calculate the $\eps$-quantization in $O(n^{2d+2}/\eps^{d+1} \log^{7d-1} \frac{n}{\eps})$ time.  

  \drieplaatjes {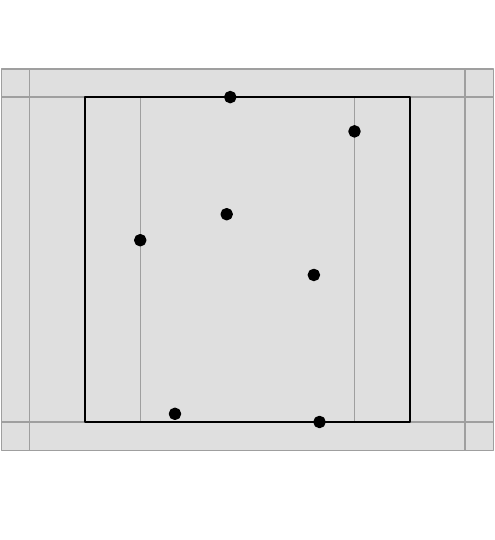} {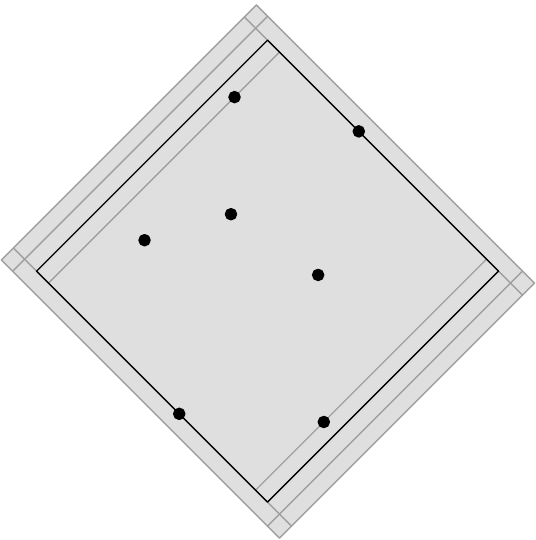} {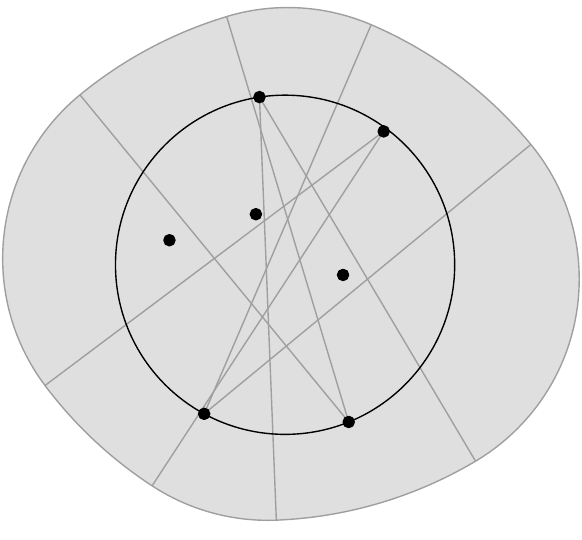}
  {\label{fig:seb} (a) Smallest enclosing ball, $L_\infty$ metric. (b) Smallest enclosing ball, $L_1$ metric. (c) Smallest enclosing ball, $L_2$ metric. The curves are circular arcs of two different radii.}

For $\sf{seb}_2$, $(\b{R}^d, \c{A}_{f,n})$ has infinite VC-dimension, but $(\b{R}^2, \c{W}_{f,n})$ has VC-dimension $\leq 9$ because it is the intersection of $2$ halfspaces and one disc.  
Any shape from $\c{A}_{f,n}$ can be formed from the disjoint union of $2n$ wedges.  
Choosing a point in the convex hull of the $n-1$ points describing $\c{A}_{f,n}$ as the vertex of the wedges will ensure that each wedge is completely inside the ball that defines part of its boundary.
Thus, in $\b{R}^2$ the $n$ $\frac{\eps}{n}$-samples of each $(\mu_{p_i}, \c{A}_{f,n})$ are of size $\alpha_f(n,\eps) = O(n^4/\eps^2 \log \frac{n}{\eps})$ and can all be calculated in $O(n^{5}/\eps^{2} \log^2 \frac{n}{\eps})$ time.  And then the $\eps$-quantization can be calculated in $O(n^{16.5}/\eps^7 \log^{3.5} \frac{n}{\eps})$ time, using range searching data structures.

For $\sf{seb}_2$, in $\b{R}^d$ for $d>3$, we can form shapes from $\c{A}_{f,n}$ with disjoint unions of generalized cone from a family $\c{W}_{f,n}$, where $(\b{R}^d, \c{W}_{f,n})$ has shatter dimension $O(d^2)$.  We need $O(n^{\lfloor d/2 \rfloor})$ such shapes from $\c{W}_{f,n}$ to form one shape $A \in \c{A}_{f,n}$, because $A$ has boundary described by $O(n^{\lfloor d/2 \rfloor})$ sphere piece with one of $d$ different radii.  
The VC-dimension of each $(\mu_{p_i}, \c{W}_{f,n})$ is $O(d^2 \log d)$ in $\b{R}^d$, and we can create $n$ $\frac{\eps}{n}$-quantization of each $(\mu_{p_i}, \c{A}_{f,n})$ of size $\alpha_f(n,\eps) = O(n^{2+2\lfloor d/2 \rfloor} / \eps^2 \log \frac{n}{\eps})$ in $O(n^{3+2\lfloor d/2 \rfloor} /\eps^2 \log^2 \frac{n}{\eps})$ total time.
Then the $\eps$-quantization can be computed in $O((n^{d+3}/\eps^2)^{d+2 - 1/d})$ time (ignoring boundary cases with floor operations) using a range searching data structure.

\paragraph{Diameter.}
Figure \ref{fig:diam} shows an example of $\c{A}_{f,n}$ for the diameter of a point set in $\b{R}^2$.  Here $(\b{R}^d, \c{A}_{f,n})$ has infinite VC-dimension.  It is formed by the intersection of balls of the same radius centered at the points.  Thus a shape from $\c{A}_{f,n}$ is determined by at most $n$ balls, and since they are each the same radius, we can construct a shape from $\c{A}_{f,n}$ from the disjoint union of $n$ wedges, as with \sf{seb}$_2$.  And since each wedge is the intersection of $d$ halfspaces and 1 disc, $(\b{R}^d, \c{W}_{f,n})$ has VC-dimension $(d+1)^2$.  
Thus we can construct $n$ $\frac{\eps}{n}$-samples for each $(\mu_{p_i}, \c{A}_{f,n})$ of size $\alpha_f(n,\eps) = O(n^4/\eps^2 \log \frac{n}{\eps})$ in total time $O(n^5 / \eps^2 \log^2 \frac{n}{\eps})$.  
However, given a set of $n$ points, the shape which defines the set of points where the diameter will not increase has complexity $O(n)$.  This is the family $\c{A}_{f,n+1}$, the union of $n$ balls.  This implies the size of the basis $\sigma_f$ used in Algorithm \ref{alg:count-prob} is $n$.   Hence, it takes time $O((n^5/\eps^2 \log \frac{n}{\eps})^{n+1})$ to construct the $\eps$-quantization.  

  \eenplaatje {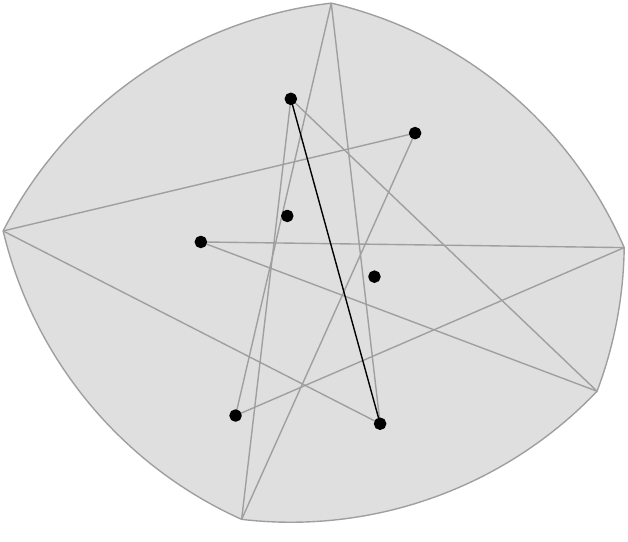}
  {\label{fig:diam} Diameter. The curves are circular arcs all of the same radius.}

\paragraph{Convex hull.}
Figure \ref{fig:ch} shows examples of $\c{A}_{f,n}$ for the convex hull, measured either by area (\sf{cha}) or perimeter (\sf{chp}) in $\b{R}^2$. For both, $(\b{R}^2, \c{A}_{f,n})$ has infinite VC-dimension.  
For \sf{cha} $(\b{R}^2, \c{W}_{f,n})$ has VC-dimension $7$, because wedges are triangles.  In higher dimensions \sf{cha} can continue to use wedges, but needs $O(n^{\lfloor d/2 \rfloor})$ of them.
For \sf{chp}, the wedges boundary is described by $d$ hyperplanes and an ellipse boundary part. We cannot guarantee that the intersection of all of these parts describes the wedge because the ellipse may be too small and may cut off part of the intersection of halfspaces.  But in $\b{R}^2$ the wedge clearly does have shatter dimension $4 + 5$, so the VC-dimension of $(\b{R}^2, \c{W}_{f,n})$ is $O(1)$.  In higher dimensions we can use generalized cone shapes with VC-dimension $O(d^2 \log d)$  and we may need $O(n^{\lfloor d/2 \rfloor})$ of them.

  \tweeplaatjes {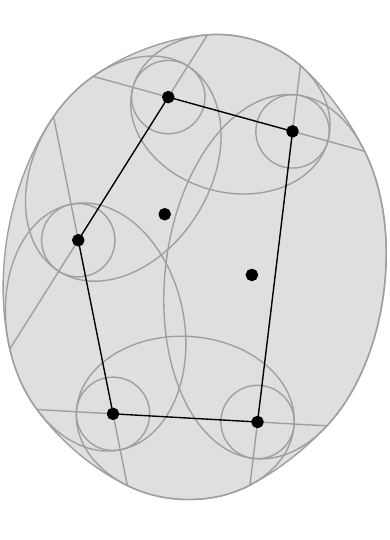} {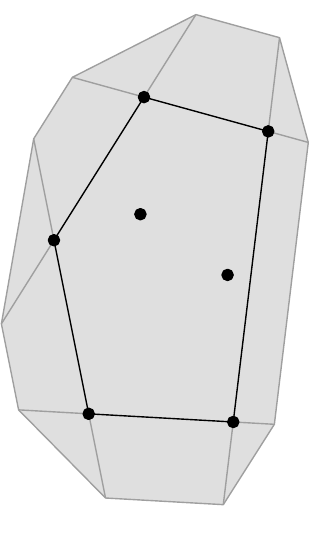}
  {\label{fig:ch} (a) Convex hull, measured by perimeter. The curves are ellipse parts. (b) Convex hull, measured by area.}

For both \sf{cha} and \sf{chp} we can calculate $n$ $\frac{\eps}{n}$-samples for each $(\mu_{p_i}, \c{A}_{f,n})$ of size $\alpha_f(n,\eps) = O(n^{2+2\lfloor d/2 \rfloor} / \eps^2 \log \frac{n}{\eps})$ in $O(n^{3+2\lfloor d/2 \rfloor} /\eps^2 \log^2 \frac{n}{\eps})$ total time.
However, given a set of $n$ points, the shape which defines the set of points where the convex hull will not increase has complexity $O(n)$.  This is the family $\c{A}_{f,n+1}$.  Like \sf{diam}, this implies the size of the basis $\sigma_f$ used in Algorithm \ref{alg:count-prob} is $n$.   Hence, it takes time $O((n^{d+3}/\eps^2 \log \frac{n}{\eps})^{n+1})$ to construct the $\eps$-quantizations.  

 \begin{table}[h!!t]
\caption{$\eps$-Samples for Summarizing Shape Family $\c{A}_{f,n}$.}
\small
\centering
\begin{tabular}{|l|c|c|c|c|c|c|}
\hline 
abbrv. & $\alpha_f(n,\eps)$ & $\eta = n \alpha_f(n,\eps)$ & $\nu_f$ & $\eta^{\nu_f}$ & $\sf{RS}_f(n,\eps)$ & runtime
\\ \hline \hline
\sf{aabbp} & $\tilde{O}(\frac{n^2}{\eps^2})$ & $\tilde{O}(n^3 / \eps^2)$ & $2d$ & $\tilde{O}(n^{6d}/\eps^{4d})$ & $\tilde{O}(1)$ & $\tilde{O}(n^{6d}/\eps^{4d})$
\\ \hline
\sf{aabba} & $\tilde{O}(\frac{n^2}{\eps^2})$ & $\tilde{O}(n^3 / \eps^2)$ & $2d$ & $\tilde{O}(n^{6d}/\eps^{4d})$ & $\tilde{O}(1)$ & $\tilde{O}(n^{6d}/\eps^{4d})$
\\ \hline
\sf{seb$_\infty$} & $\tilde{O}(\frac{n}{\eps})$ & $\tilde{O}(n^2 / \eps)$ & $d+1$ & $\tilde{O}(n^{2d+2}/\eps^{d+1})$ & $\tilde{O}(1)$ & $\tilde{O}(n^{2d+2}/\eps^{d+1})$
\\ \hline
\sf{seb$_1$} & $\tilde{O}(\frac{n}{\eps})$ & $\tilde{O}(n^2 / \eps)$ & $d+1$ & $\tilde{O}(n^{2d+2}/\eps^{d+1})$ & $\tilde{O}(1)$ & $\tilde{O}(n^{2d+2}/\eps^{d+1})$
\\ \hline
\sf{seb$_2$} & $\tilde{O}(\frac{n^{d+2}}{\eps^2})$ & $\tilde{O}(n^{d+3}/\eps^2)$ & $d+1$ & $\tilde{O}((n^{d+3} / \eps^{2})^{d+1})$ & $O((n^{d+3}/\eps^2)^{1-1/d})$ & $\tilde{O}((n^{d+3}/\eps^2)^{d+2-1/d})$
\\ \hline
\sf{diam} & $\tilde{O}(\frac{n^4}{\eps^2})$ & $\tilde{O}(n^5/\eps^2)$ & $n$ & $\tilde{O}(n^5/\eps^2)^n$ & $O(n^5 / \eps^2)$ & $\tilde{O}((n^5/\eps^2)^{n+1})$
\\ \hline
\sf{cha} &  $\tilde{O}(\frac{n^{d+2}}{\eps^2})$ & $\tilde{O}(n^{d+3}/\eps^2)$ & $n$ & $\tilde{O}((n^{d+3}/\eps^2)^n)$ & $O(n^{d+3}/\eps^2)$ & $\tilde{O}((n^{d+3}/\eps^2)^{n+1})$
\\ \hline
\sf{chp} &  $\tilde{O}(\frac{n^{d+2}}{\eps^2})$ & $\tilde{O}(n^{d+3}/\eps^2)$ & $n$ & $\tilde{O}((n^{d+3}/\eps^2)^n)$ & $O(n^{d+3}/\eps^2)$ & $\tilde{O}((n^{d+3}/\eps^2)^{n+1})$
\\ \hline
\end{tabular}
\label{tbl:Afn-size}
\\  $\tilde{O}(f(n,\eps))$ ignores poly-logarithmic factors $(\log \frac{n}{\eps})^{O(\textrm{poly}(d))}$.
\end{table}

\end{document}